\documentclass[12pt,draftclsnofoot,onecolumn]{IEEEtran}
\usepackage{stfloats}
\usepackage{makecell}
\usepackage{graphicx}
\usepackage[cmex10]{amsmath}
\usepackage{cases}
\usepackage[tight,footnotesize]{subfigure}
\usepackage{amsthm} %proof
\usepackage{cite}
\usepackage{citesort}
\usepackage{amssymb}
\allowdisplaybreaks[4]
\usepackage{algorithm}
\usepackage{algorithmic}
\usepackage{multirow}
\usepackage{amsmath}
\usepackage{xcolor}
\usepackage{CJK}
\usepackage{subeqnarray}
\usepackage{cases}
\usepackage{enumerate}

\usepackage{stfloats}

\newtheorem{lemma}{\hskip\parindent\bf{Lemma}}%[section]

\newtheorem{theorem}{\hskip\parindent\bf{Theorem}}

\ifCLASSINFOpdf
\else
\fi

\begin{document}

\title{Edge and Central Cloud Computing: A Perfect Pairing for High Energy Efficiency and Low-latency}

\author{Xiaoyan Hu,~\IEEEmembership{Student Member,~IEEE,} Lifeng Wang,~\IEEEmembership{Member,~IEEE,}  Kai-Kit~Wong,~\IEEEmembership{Fellow,~IEEE,} %\\
Meixia Tao,~\IEEEmembership{Fellow,~IEEE,} } %\\
%Yangyang Zhang,  and Zhongbin Zheng
%\thanks{X. Hu, L. Wang and K.-K. Wong are with the Department of Electronic and
%Electrical Engineering, University College London, London  WC1E 7JE, UK (Email: \{xiaoyan.hu.16, lifeng.wang, kai-kit.wong\}@ucl.ac.uk).}
%\thanks{M. Tao is with the Department of Electronic Engineering, Shanghai Jiao Tong University, Shanghai 200240, China (Email: mxtao@sjtu.edu.cn).}
%\thanks{Y. Zhang is with the  Kuang-Chi Institute of Advanced Technology, Shenzhen 518057, China (Email: yangyang.zhang@kuang-chi.org).}
%\thanks{Z. Zheng is with the East China Institute of Telecommunications, China Academy of Information and Communications Technology, Shanghai 200001, China (Email: ben@ecit.org.cn).}
%}
%(Email: $\rm{mxtao}@sjtu.edu.cn$).
\maketitle
\vspace{-1.5cm}
\begin{abstract}
In this paper, we study the coexistence and synergy between edge and central cloud computing in a heterogeneous cellular network (HetNet), which contains a multi-antenna macro base station (MBS), multiple multi-antenna small base stations (SBSs) and multiple single-antenna user equipment (UEs). The SBSs are empowered by edge clouds offering limited computing services for UEs, whereas the MBS provides  high-performance central cloud computing services to UEs via a restricted multiple-input multiple-output (MIMO) backhaul to their associated SBSs. With processing latency constraints at  the central and  edge networks, we aim to minimize the system energy consumption used for task offloading and computation. The problem is formulated by jointly optimizing the cloud selection, the UEs' transmit powers, the SBSs' receive beamformers, and the SBSs' transmit covariance matrices, which is {a mixed-integer and non-convex optimization problem}. Based on methods such as decomposition approach and successive pseudoconvex approach, a tractable solution is proposed via an iterative algorithm.  The simulation results show that our proposed solution can achieve great performance gain over conventional schemes using edge or central cloud alone. Also, with large-scale antennas at the MBS, the massive MIMO backhaul can significantly reduce the complexity of the proposed algorithm and obtain even better performance.
\end{abstract}
\vspace{-0.2 cm}

\IEEEpeerreviewmaketitle
\newpage
\section{Introduction}\label{sec:Introduction}
%\subsection{Motivation and Prior Works}\label{sec:PW}
\subsection{Motivations}\label{sec:PW}
{To deal with the computation-intensive tasks generated by a large variety of mobile applications, the concept of central cloud computing (CCC) first emerges,
which offloads these tasks to remote powerful computing centers, also known as central cloud.}
Mobile cloud computing (MCC) integrates CCC into mobile environment to facilitate mobile users to take full advantages of cloud resources \cite{S_A.U.R.Khan14ASurvey,J_H.T.Dinh13Asurvey}. %MCC is the mainstream for computing before the emergence of MEC.
Although CCC/MCC can provide high-performance computing services for UEs, it has one inherent disadvantage, i.e., {the central cloud is usually located} far away from the UEs. Hence, accessing the CCC/MCC services induces excessive transmission latency, which aggravates the backhaul congestion. Besides, it is easy to encounter the performance bottleneck  considering the finite backhaul capacity and exponentially growing mobile data, which has led to the emergence of edge (cloud) computing.

Edge (cloud) computing has recently been regarded as one of the key enablers to shape the future intelligent wireless networks. The rationale behind edge computing is that cloud computing can be carried out at the edge of wireless networks {which is close to UEs,} so as to facilitate computation offloading of UEs and prolong their battery lifetime~\cite{J_S.Barbarossa14Communicating,J_Y.Mao17Mobile}.
The standardization bodies and industry
associations such as  ETSI and 5GAA have identified various edge computing use cases  for 5G cellular networks, such as vehicle-to-everything (V2X) and massive machine-type communications (mMTC), etc.,~\cite{ETSI_2018_white_paper,5GAA-VISUAL}.
For practical deployment, several edge computing architectures have already been proposed, such as mobile edge
computing (MEC)~\cite{P_Mach_2017}, fog computing~\cite{M_chiang_2017,K_Dolui_2017}, and also cloudlets~\cite{cloudlets_2012}.
%MEC allows base stations (BSs) to have the ability of storage and processing, to guarantee that UEs are directly connected to the edge clouds. Fog computing  is a more flexible computing architecture consisting of highly heterogeneous fog computing nodes with different levels of computing ability such as routers and network gateways. In wireless local area networks (Wi-Fi access), cloudlets run  virtual machines  and the computing resources allocated to cloudlets are managed by cloudlet agents~\cite{cloudlets_2012}.
%Recently, multi-access edge computing (also using the same acronym ``MEC'' originated from mobile edge computing) has been introduced to support multiple access technologies including cellular, Wi-Fi, etc.~\cite{ETSI_2017_MEC} %(FCNs)

%Since Internet-of-Things (IoT) devices may lack computing capability, edge computing can achieve local execution, which avoids the frequent delivery of massive computing tasks to the core networks with central cloud for computing, and thus reduces the computing latency and backhaul congestion \cite{mungchiang_2016_dec,X_Lyu_2017}.
%The survey work~\cite{mungchiang_2016_dec} presents a comprehensive overview of fog computing in IoT networks and illustrates how fog computing tackles the challenges in IoT. In~\cite{X_Lyu_2017}, Lyapunov optimization techniques are adopted to develop an online MEC scheduling solution with partial knowledge of IoT network.

Edge computing is of great benefit to resource-limited UEs, e.g., the Internet-of-Things (IoT) devices, which avoids
the frequent delivery of massive computing tasks to the core networks with central cloud for computing, and thus reduces the
transmission latency and backhaul congestion \cite{mungchiang_2016_dec,X_Lyu_2017}.
However, the computing capabilities at the edge servers/clouds are also limited in general due to the cost and their constrained size. For UEs with highly computation-intensive  tasks, the edge computing servers/clouds may be incapable of providing them with satisfactory computing services. Under this situation, CCC/MCC has been shown to be an effective solution. Hence, in order to improve the quality of service (QoS) for dealing with a wide range of UEs' computation tasks, applying the architecture with the coexistence and cooperation between the edge and central clouds could be a promising option.

\subsection{Related Works}\label{sec:RW}
Extensive works on CCC/MCC have been conducted to explore {the potential} of central cloud. Several system architectures with various code offloading frameworks have been studied, e.g., MAUI \cite{C_E.Cuervo10MAUI} and ThinkAir \cite{C_S.Kosta12Thinkair}.  In \cite{J_Z.Xiao13Dynamic}, dynamic resource allocation using virtualization technology was studied to achieve overload avoidance and  green computing  by minimizing the number of physical machines.
%The performance and power management \cite{C_H.N.Van10Performance}
Also, a computation offloading algorithm was proposed in \cite{J_S.Deng15Computation}  to deal with multiple services in workflow by leveraging MCC.

Recently, considerable attention has been paid to the design and analysis of edge computing in cellular networks,
e.g.,~\cite{Olga_Mun_2015,chenxu_2016,Changsheng_you_mar_2017,A_o_simeone,Cheng_wang_aug_2017}. The tradeoff between
energy consumption and latency in information transmission and computation was analyzed in \cite{Olga_Mun_2015}, where
an UE offloaded its application tasks to a small BS (SBS) for processing. In~\cite{chenxu_2016}, a multi-user
computation offloading problem was considered in a single-cell scenario and game-theoretical solutions were proposed in order to maximize the cell load and minimize the cost in terms of computational time and energy simultaneously. % in order to maximize the cell load and minimize the cost in terms of computational time and energy {simultaneously}.
Later in~\cite{Changsheng_you_mar_2017}, time and frequency allocation problems for improving energy efficiency were studied by considering multi-user computation offloading in a single cell equipped with limited cloud capacity, where an offloading priority function was derived to accommodate users' priorities. % equipped with limited cloud capacity. %, where an offloading priority function is derived to accommodate users' priorities.
The work of~\cite{A_o_simeone} examined a single-cloudlet scenario where multiple UEs were served with equal-time sharing, and a successive convex optimization approach was developed to minimize the network energy consumption under a computing latency constraint.
Recent works related to edge computing also focus on multi-service scenarios. % where BSs are capable of computing and caching
For example, \cite{Cheng_wang_aug_2017} considered a single MEC server with storage capability and attempted to maximize the revenue of providing both the computing and caching services.
{Besides, the effects of implementing the edge computing in energy harvesting networks have been verified in \cite{X_Sun_2017,C_You_2016,Xiaoyan_hu_2018,Jiexu_WPT_MEC2018,J_S.Bi18Computation},
where the lifetime of UE's battery could be further prolonged..}

The complementary benefits between edge and central cloud  have driven {research} towards the coexistence and cooperation between edge and central clouds \cite{J_L.Lei13Challenges}. % \cite{C_F.BenJemaa16QoS,J_M.Villari16Osmotic,C_S.Shekhar17Dynamic}.
One such example was \cite{C_T.Zhao15ACooperative} where a delay-aware scheduling between local and Internet clouds was studied, and a priority-based cooperation policy was given to maximize the total successful offloading probability.
The placement and provisioning of virtualized network functions were explored in  \cite{C_F.BenJemaa16QoS}, where {a QoS}-aware optimization strategy was proposed over an edge-central carrier cloud infrastructure. Also, the work in \cite{C_M.Chen16Joint} considered that an edge server and a central cloud coexist to complete the UEs' computations cooperatively, where a wired connection was assumed between the edge and the central cloud.

\subsection{Our Contributions}\label{sec:Contributions}
Most of the existing computing works focused on either the edge or central cloud computing independently, and the edge computing works mainly concentrated on small-scale networks such as the single MEC server or cloudlet case \cite{chenxu_2016,Changsheng_you_mar_2017,A_o_simeone,Cheng_wang_aug_2017,X_Lyu_2017,C_You_2016,Xiaoyan_hu_2018,Jiexu_WPT_MEC2018,J_S.Bi18Computation}.
Even though edge computing has been regarded as a promising trend to deal with the ever growing mobile computing data,
it cannot entirely replace the present central cloud computing, due to the fact that edge computing is set to push limited processing and storage capabilities close to UEs but may be incapable of dealing with big data processing. {The latest white paper from ETSI has further illustrated that central cloud computing and edge computing are highly complementary and significant benefits can be attained when utilizing them both~\cite{ETSI_2018_white_paper}.}
%However, the  architecture with the coexistence and cooperation of edge and central cloud has not been thoroughly studied, especially from the communication perspective \cite{J_Y.Mao17Mobile}.
{However, the existing works \cite{C_T.Zhao15ACooperative,C_F.BenJemaa16QoS,C_M.Chen16Joint} considering the co-existence of edge and central cloud computing  either focus on delay-aware priority scheduling, virtualized resource allocation, or offloading with wired backhaul. The issues related to offloading decision and resource allocation of hybrid edge/central cloud computing networks with wireless backhaul  have not been thoroughly studied, especially from the viewpoint of communication \cite{J_Y.Mao17Mobile}.}
{Therefore, this paper studies the deployment of heterogeneous edge and central clouds to leverage the easy access of edge clouds and the abundant computing resources at the central cloud, mainly from the viewpoint of communication by considering cloud selection and resource allocation.
To our best knowledge, this is the first work addressing the integrated edge and central cloud computing in heterogeneous cellular networks (HetNets) while considering the physical properties of wireless backhaul.}

Our main contributions are summarized as follows:
\begin{itemize}
  \item \textbf{Hybrid Edge/Central Cloud Computing Architecture:}
   We consider a hybrid edge and central cloud computing architecture in a two-tier HetNet, including one macro cell with a macro BS (MBS) and multiple small cells each with an SBS. The edge clouds with limited computing capabilities are co-located at or linked to the SBSs by error-free optical fibers while the central cloud with ultra-high computing capability is connected with the MBS through optical fibers as well. The UEs can offload their computation tasks directly to the SBSs to access the edge cloud computing services (edge computing mode) or further offload to the MBS through the restricted  multiple-input multiple-output (MIMO)/massive MIMO backhaul to utilize the central cloud computing services (central computing mode). Cooperation of edge and central clouds will improve the QoS and ensure the scalability and load balancing between the edge and central clouds.

 \item \textbf{Problem Formulation with Joint Optimization on the Cloud Selection, Access Transmit Powers, Receive Beamforming Vectors and Backhaul Transmit Covariance Matrices:}
     Our aim is to minimize the network's energy consumption for task offloading and computation under both the central and edge processing latency constraints through jointly optimizing the cloud selection, the UEs' transmit powers, the SBSs' receive beamforming vectors, and the SBSs' transmit covariance matrices.
     The central processing latency constraint requires the {backhaul} transmission latency being lower than the computing latency at the edge cloud; otherwise, the central cloud will not be selected.
     The edge processing latency \mbox{constraint} requires the corresponding latency not exceeding a targeted threshold to guarantee the quality of services provided by edge cloud.
     A mixed-integer and non-convex optimization problem is formulated accordingly, which is {NP-hard} in general. For the case of massive MIMO backhaul, we consider two low-complexity linear processing methods, namely maximal-ratio
     combining (MRC) and zero-forcing (ZF), and the corresponding optimization problems can be much simplified.

\item \textbf{Algorithm Design:} An iterative algorithm is \mbox{developed} to solve the  combinatorial mixed-integer and non-convex optimization problem {corresponding to} the case with traditional MIMO backhaul. In particular, we show that in each iteration, the UEs' transmit powers and the SBSs' receive beamforming vectors can be optimized in closed-form, and the SBSs' transmit covariance matrix solution is obtained by leveraging a successive pseudoconvex optimization approach. In addition, the massive MIMO backhaul solutions can be easily obtained thanks to the unique features of massive MIMO transmission, which significantly reduce the complexity of the algorithm. {The practicality of the proposed algorithm lies in that it can properly address the issues of cloud decision and resource allocation for a HetNet architecture with hybrid edge/central cloud computing resources while considering the physical properties of wireless backhaul.}

\item {\textbf{Design Insights:}} Simulation results are presented to demonstrate the efficiency of the proposed algorithm and shed light on the effects of key parameters such as the offloaded task size, edge processing latency threshold, and edge cloud's CPU frequency.
    It is confirmed that the solution of the integrated edge and central cloud computing scheme proposed in this work can achieve better performance than the schemes with edge (cloud) computing alone or central cloud computing alone, and outperforms all the other benchmark solutions. In addition, low-complexity massive MIMO solution with ZF receiver could  always  outperform the solution with traditional  MIMO backhaul,  while the solution with  MRC receiver could achieve similar or better performance than the traditional  MIMO one in certain scenarios.
\end{itemize}

The rest of this paper is organized as follows. In Section~\ref{sec:system}, the considered system model is described and
the corresponding optimization problem is formulated. The proposed algorithm under traditional MIMO backhaul is presented in
Section~\ref{algorithm_design}, and massive MIMO backhaul solution is given in Section \ref{sec:Special_MEC_ZF}.
Section~\ref{sec:simulation} provides the simulation results. Finally, we have some concluding remarks in Section~\ref{sec:conclusion}.

{\em Notations}---In this paper, the upper and lower
case bold symbols denote matrices and vectors, respectively. The notations $(\cdot)^H$ and $(\cdot)^\dagger$ are conjugate transpose and conjugate operators for vectors or matrices, respectively.  $\left[x\right]^+$ is used as $\max\left\{x,0\right\}$. In addition, $\det\left(\mathbf{A}\right)$ denotes
the determinant of $\mathbf{A}$, and $\mathrm{tr}\left\{\mathbf{A}\right\}$ is the trace of $\mathbf{A}$. Also,
${\rm eig}\left\{\mathbf{A}\right\}$ denotes the set of all the eigenvalues for $\mathbf{A}$, and ${\rm eigvec}\left\{\cdot\right\}$ gives the eigenvector for a given eigenvalue of $\mathbf{A}$.
$\langle\mathbf{A}_1,\mathbf{A}_2\rangle\triangleq\mathfrak{R}\{\mathrm{tr}(\mathbf{A}_1^H\mathbf{A}_2)\}$, where
$\mathfrak{R}\{\cdot\}$ is the real-value operator, and $\nabla_{\mathbf{A}}f(\mathbf{A})$ denotes the Jacobian matrix of $f(\mathbf{A})$ with respect to (w.r.t.) $\mathbf{A}$.

\section{System Model and Problem Formulation}\label{sec:system}
\begin{figure}[t!]
\centering
\includegraphics[width=3 in]{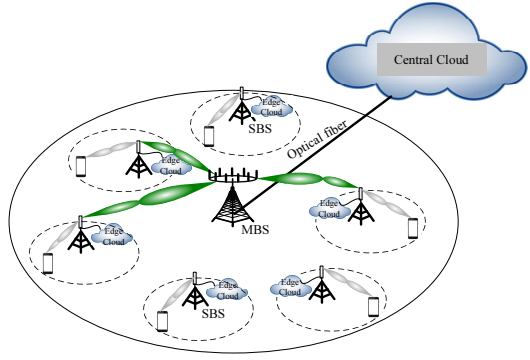}
\caption{An illustration of two-tier HetNets equipped with edge clouds  associated with  the SBSs and central cloud connected by the MBS via optical fiber, where the MBS provides central cloud computing services for UEs through restricted MIMO/massive MIMO backhaul for addressing more complicated computing tasks which cannot be handled by the SBSs' edge cloud  due to the limited computing capability.}% {edge clouds  at or linked to the SBSs through optical fiber}
\label{cloud_edge}
\end{figure}
As shown in Fig. \ref{cloud_edge}, we consider a two-tier HetNet, where  an $M$-antenna MBS provides wireless MIMO backhaul  and is fiber-optic connected to the central cloud with super computing capability, and $N$ SBSs with edge clouds can provide limited computing capabilities.\footnote{The central cloud can be regarded as the computing part of the cloud radio access network (Cloud RAN) \cite{J_D.Pompili15Dynamic}. The edge cloud can be an independent edge computing server co-located at the corresponding SBS or a certain part of computing capability  allocated to the SBS from a fiber-optic connected  edge computing center \cite{J_Y.Mao17Mobile}.}
In each small cell, an SBS equipped with $L$ antennas serves a single-antenna UE\footnote{The case of serving multiple UEs in each small cell can be effectively dealt with by using the  existing orthogonal multiple access techniques for radio resource allocation. { In addition, our work can be viewed as the equal computing resource sharing case for multiple active UEs in a small cell, or dedicated computing resource policy case for different types of edge computing service, i.e., each service will be granted one dedicated computing resource.}}. Note that existing user association schemes~\cite{dantong-survey} can be adopted to determine which user is connected to an SBS. Let $\mathcal{N}=\{1,\dots,N\}$ denote the set of the SBSs as well as the UEs.
Each UE $n\in\mathcal{N}$ has an atomic highly
integrated computation-intensive task, characterized by a positive tuple $[I_n, O_n, K_n]$, which cannot be partitioned for parallel execution. Here $I_n$ is the size (in bits) of the computation task-input data (e.g., the program codes and the input parameters) which cannot be divided and has to be offloaded as  a whole for computation, $O_n$ denotes the size of the task-output data corresponding to the results of the $I_n$ input data,  and $K_n$ is the amount of required computing resource for computing 1-bit of the input data (i.e., the number of CPU cycles required).\footnote{The parameters in the task tuple of  $[I_n, O_n, K_n]$ can be obtained through  task profilers by applying the methods (e.g., call graph analysis) \cite{J_Y.Mao17Mobile,C_E.Cuervo10MAUI,C_A.P.Miettinen10Energy,C_S.Melendez17Computation,J_L.Yang13AFramework}. It is assumed that the size of computing outputs, i.e., $O_n$ (usually a few command bits) is much smaller than $I_n$ (usually measured by Mbit) in practice, and thus the downlink overhead such as time and energy consumption for delivering the output data back  to the UEs is negligible and can be ignored.}   Let $B^\mathrm{a}$ and $B^\mathrm{b}$ denote the bandwidths allocated to UEs' access links to their serving SBSs and SBSs' backhaul links to the MBS, respectively. A coordination and monitoring protocol between SBSs and MBS, as the one used in \cite{J_T.Han16Small,J_D.Valocchi17SigMA}, is needed.

Assuming that the UEs are endowed with very limited computing resources, they tend to choose computation offloading to complete their computation tasks remotely, so as to save their own energy and resources. Since the computation tasks offloaded by the UEs could be executed either at the edge cloud or central cloud, the cloud selection needs
to be appropriately determined before evaluating the computation latency and energy consumption. Let the binary indicator
$c_n$ denote the computing decision, where $c_n=1$ indicates edge computing, and $c_n=0$ indicates central cloud computing being selected for each UE $n \in \mathcal{N}$.
In the sequel, we will study the latency and energy consumption of the network, and then formulate the optimization problem for
minimizing the network's total energy consumption for task offloading and computation under the central  and edge processing latency constraints.

\subsection{Transmission and Computing Latency}\label{sec:CTM}
\subsubsection{Access Transmission Latency} The uplink transmission rate of UE $n$ for offloading the $I_n$-bit computation tasks to its serving SBS is expressed as
\begin{align}\label{eq:R1}
R_n^{\mathrm{a}}(\mathbf{p}^{\mathrm{u}},\mathbf{w}_n)=B^\mathrm{a}\log_2
\left(1+\gamma_n^{\mathrm{a}}(\mathbf{p}^{\mathrm{u}},\mathbf{w}_n)\right), \ n\in\mathcal{N}
\end{align}
with the signal-to-interference-plus-noise ratio (SINR)
\begin{equation}\label{eq:SINR1}
\gamma_n^{\mathrm{a}}(\mathbf{p}^{\mathrm{u}},\mathbf{w}_n)=\frac{p_n^{\mathrm{u}}|\mathbf{w}_n^H{\mathbf{h}_{n,n}^{\rm a}}|^2}
{\sum_{i=1,i\neq n}^{N} p_i^{\mathrm{u}}|\mathbf{w}_n^H {\mathbf{h}_{i,n}^{\rm a}}|^2+|\mathbf{w}_n^H\mathbf{n}_n|^2},
\end{equation}
where $\mathbf{w}_n$ is the receive beamforming vector of the $n$-th SBS, ${\mathbf{h}_{i,n}^{\rm a}}\in\mathbb{C}^{L\times1}$ is the access channel vector
between UE $i$ and SBS $n$, $\mathbf{n}_n$ is a vector of the additive white Gaussian noise with zero mean and variance $\sigma_n^2$, and
$\mathbf{p}^{\mathrm{u}}\triangleq[p_1^{\mathrm{u}},\dots,p_N^{\mathrm{u}}]^T\in\mathbb{R}^{N\times1}$ denotes the transmit power vector
of the UEs. Therefore,  the uplink access transmission latency for UE $n$'s task offloading can be calculated as
%given an arbitrary offloaded computation task size of the $n$-th UE's, denoted as $I_n$ (bits),
%$T_n^\mathrm{a}(\mathbf{p}^{\mathrm{u}},\mathbf{w}_n)=\frac{I_n}{R_n^{\mathrm{a}}(\mathbf{p}^{\mathrm{u}},\mathbf{w}_n)}$.
\begin{align}\label{time_access}
T_n^\mathrm{a}(\mathbf{p}^{\mathrm{u}},\mathbf{w}_n)=\frac{I_n}{R_n^{\mathrm{a}}(\mathbf{p}^{\mathrm{u}},\mathbf{w}_n)}, \ n\in\mathcal{N}.
\end{align}

\subsubsection{Edge Computing Latency ($c_n=1$)}\label{edge_cloud_latency}
 Let $f_n$  denote the CPU clock frequency of the $n$-th edge cloud associated with  SBS $n$, and thus the corresponding edge computation latency for dealing with the $I_n$-bits input data can be described as % at the $n$-th SBS %$T_n^{\mathrm{edge}}=\frac{\vartheta I_n}{f}$.
\begin{align}\label{time_edge}
T_n^{\mathrm{edge}}=\frac{ I_nK_n}{f_n}, \ n\in\mathcal{N},
\end{align}
which indicates that the value of edge computing latency heavily depends on the offloaded task size, the unit computing resource required  and edge cloud's CPU clock frequency.
%\subsubsection{Backhaul Transmission Latency ($c_n=0$)}
\subsubsection{Central Cloud Processing Latency/Backhaul Transmission Latency ($c_n=0$)}\label{cloud_time-process}
The central cloud processing latency results from backhaul transmission
and task execution at the central cloud. Due to the central cloud's super computing capability, its computing time is much lower
than edge computing, thus we assume that the time for central cloud computing is negligible. Hence, the central cloud processing latency, i.e.,
the backhaul transmission latency for the $n$-th UE can be  calculated as\footnote{{In our considered scenario, the accessing latency of MBS to the central cloud through optical fiber should be negligible especially compared with the wireless backhaul transmission latency. For the extreme case that the optical fiber transmission latency is not negligible, the central cloud processing latency can be re-expressed as $T_n^{\mathrm{central}}(\mathbf{Q})=\frac{I_n}{R_n^{\mathrm{b}}(\mathbf{Q})}+T{_{\mathrm{of}}^\mathrm{central}}$, where $T{_{\mathrm{of}}^\mathrm{central}}$ is a maximum threshold of optical fiber transmission latency. Even though, the proposed algorithms are still effective.}}
\begin{align}\label{time_backhaul}
T_n^{\mathrm{central}}(\mathbf{Q})=\frac{I_n}{R_n^{\mathrm{b}}(\mathbf{Q})}, \ n\in\mathcal{N},
\end{align}
where $R_n^{\mathrm{b}}(\mathbf{Q})$ is the backhaul transmission rate given by
\begin{align}\label{backhaul_rate}
\hspace{-0.3 cm} R_n^{\mathrm{b}}(\mathbf{Q})=B^\mathrm{b} \log_2\det\left(\mathbf{I}+\Psi(\mathbf{Q}_{-n})^{-1}
\mathbf{H}_n^\mathrm{b} \mathbf{Q}_n \left(\mathbf{H}_n^\mathrm{b}\right)^H\right),
\end{align}
with the noise-plus-interference covariance matrix %denoted as
$\Psi(\mathbf{Q}_{-n})=\sigma^2\mathbf{I}+\sum_{i=1,i\neq n}^N\mathbf{H}_i^\mathrm{b} \mathbf{Q}_i \left(\mathbf{H}_i^\mathrm{b}\right)^H$.
In \eqref{backhaul_rate}, $\mathbf{Q}_n$ is the transmit covariance matrix of SBS $n$, $\mathbf{Q}=\{\mathbf{Q}_n\}_{n=1}^N$ and
$\mathbf{Q}_{-n}=\{\mathbf{Q}_i\}_{i=1,i\neq n}^N$ are respectively the compact transmit covariance matrices and the compact transmit covariance matrices except $\mathbf{Q}_n$, and $\mathbf{H}_n^\mathrm{b}\in\mathbb{C}^{M\times L}$ is the backhaul channel matrix
from SBS $n$ to the MBS.  Note that if the task of UE $n\in\mathcal{N}$ is executed by the edge cloud of SBS $n$, i.e. $c_n$ = 1, the transmit covariance matrix at SBS $n$ shall be $\mathbf{Q}_n = \mathbf{0}$.

%In addition, it is assumed that the size of computing outputs $O_n$ (usually a few command bits) is much smaller than $I_n$ (usually measured by Mbit) in practice, and thus the downlink overhead such as time and energy consumption for delivering the output data back  to the UEs is negligible and can be ignored.

\subsection{Energy Consumption}
Energy consumption results from task offloading energy and task execution/computation energy.  Based on Section~\ref{sec:CTM}, the amount of energy
consumption for UE $n\in \mathcal{N}$ to offload its computing tasks to its serving SBS can be calculated as  %$E_n^\mathrm{a}=p_n^{\mathrm{u}} T_n^\mathrm{a}(\mathbf{p}^{\mathrm{u}},\mathbf{w}_n)$.
\begin{align}\label{energy_consumption_access}
E_n^\mathrm{a}=p_n^{\mathrm{u}} T_n^\mathrm{a}(\mathbf{p}^{\mathrm{u}},\mathbf{w}_n), \ n\in \mathcal{N}.
\end{align}
If the UE $n$'s task is executed by the edge cloud associated with the  SBS, the computation energy consumption at the corresponding edge server is given by %the edge processing energy consumption, i.e.,
\begin{align}\label{energy_consumption_edge}
E_n^\mathrm{edge}= \varrho_n  I_n  K_n f_n^2, \ n\in \mathcal{N},
\end{align}
where $\varrho_n$ is the effective switched capacitance of the edge cloud $n$. Else, if the task is executed by the central cloud, we then have the central processing energy consumption, including the backhaul transmission and the computation energy consumption, which is expressed as
%$E_n^\mathrm{Cloud}=\mathrm{tr}\left(\mathbf{Q}_n\right) T_n^{\mathrm{central}}(\mathbf{Q})+\zeta E_n^\mathrm{Edge}$
\begin{align}\label{energy_consumption_cloud}
E_n^\mathrm{central}=\mathrm{tr}\left(\mathbf{Q}_n\right) T_n^{\mathrm{central}}(\mathbf{Q})+\zeta_n E_n^\mathrm{edge}, \ n\in \mathcal{N},
\end{align}
where $\zeta_n$ is the ratio of the central cloud's computation energy consumption to that of the  edge cloud $n$ for computing the same UE  $n$'s task.\footnote{$\zeta_n$ can be determined by $\varrho_n$, $f_n$, and  the effective switched capacitance and the CPU frequency of the central cloud used for computing UE $n$'s task. Different values of $\{\zeta_n, n\in \mathcal{N}\}$ represent different relationships between the computing energy consumption at  central cloud and edge clouds, and may have different effects on edge/central cloud selection and system performance.}
Thus, the total energy consumption for task offloading and computation can be calculated as\footnote{Here, the static energy consumption of UEs and
SBSs consumed by the circuit or cooling  is ignored since it has negligible effect on our design.}
\begin{align}\label{total_energy_consumption}
%\hspace{-0.3 cm}E_\mathrm{total}=& \sum\limits_{n{\rm{ = }}1}^N {p_n^{\mathrm{u}} T_n^\mathrm{a}}(\mathbf{p}^{\mathrm{u}},\mathbf{w}_n)+ \sum\limits_{n{\rm{ = }}1}^N \big(c_n \varrho  I_n  \vartheta f^2\nonumber\\
%&\qquad+\left(1-c_n\right)\mathrm{tr}\left(\mathbf{Q}_n\right) T_n^{\mathrm{central}}(\mathbf{Q})\big).
E_\mathrm{total}=& \sum\limits_{n{\rm{ = }}1}^N \left(E_n^\mathrm{a}+c_n E_n^\mathrm{edge}+\left(1-c_n\right)E_n^\mathrm{central}\right).
\end{align}

\subsection{Problem Formulation} \label{Problem_Formulation}
Our aim is to minimize network's total energy consumption used for task offloading and computation under central/backhaul and edge processing latency constraints through jointly optimizing  UEs' cloud selection decisions in $\mathbf{c}=\{c_n\}_{n=1}^N$, UEs' transmit power vector $\mathbf{p}^{\mathrm{u}}$, SBSs' receive beamformers in $\mathbf{w}=\{\mathbf{w}_n\}_{n=1}^N$ and the SBSs' transmit covariance matrices in $\mathbf{Q}$. To this end, the problem is formulated as
\begin{align}\label{P1}
&\mathop {\min }\limits_{{\mathbf{c}, \mathbf{p}^{\mathrm{u}},\mathbf{w}, \mathbf{Q} }} ~E_\mathrm{total}\qquad\qquad\qquad\qquad\qquad\qquad\qquad\qquad\\
&~~~~\mathrm{s.t.} ~~~  \mathrm{C1}:  c_n \in \left\{0,1\right\},~~\forall n \in \mathcal{N},    \nonumber\\
&~~~\qquad ~~\mathrm{C2}: \left(1-c_n\right) T_n^{\mathrm{central}}(\mathbf{Q}) \leq  \alpha T_n^{\mathrm{edge}},~~\forall n \in \mathcal{N},\nonumber\\
&~~~\qquad  ~~\mathrm{C3}: T_n^\mathrm{a}(\mathbf{p}^{\mathrm{u}},\mathbf{w}_n)+c_n T_n^{\mathrm{edge}} \leq T_\mathrm{th},~~\forall n \in \mathcal{N},  \nonumber\\
&~~~\qquad ~~\mathrm{C4}: 0\leq p_n^{\mathrm{u}} \leq P_{\max}^{\rm u},~~\forall n \in \mathcal{N},
 \nonumber\\
&~~~\qquad ~~\mathrm{C5}:  \mathbf{Q}_n\succeq\mathbf{0},~~\forall n \in \mathcal{N}.
 \nonumber
\end{align}
%\begin{align}\label{P1}
%\hspace{-5mm}&\mathop {\min }\limits_{{\mathbf{c}, \mathbf{p}^{\mathrm{u}},\mathbf{w}, \mathbf{Q} }} ~E_\mathrm{total}\qquad\qquad\qquad\qquad\qquad\qquad\qquad\qquad\\
%&~\mathrm{s.t.} ~ \mathrm{C1}:  c_n \in \left\{0,1\right\},~~\forall n \in \mathcal{N},    \nonumber\\
%&\qquad \mathrm{C2}: (1-c_n) (T_n^{\mathrm{central}}(\mathbf{Q})+{T_n^{\mathrm{pre}}} ) \leq  \alpha T_n^{\mathrm{edge}},\forall n \in \mathcal{N},\nonumber\\
%&\qquad  \mathrm{C3}: T_n^\mathrm{a}(\mathbf{p}^{\mathrm{u}},\mathbf{w}_n)+c_n T_n^{\mathrm{edge}} \leq T_\mathrm{th},~~\forall n \in \mathcal{N},  \nonumber\\
%&\qquad \mathrm{C4}: 0\leq p_n^{\mathrm{u}} \leq P_{\max}^{\rm u},~~\forall n \in \mathcal{N},
% \nonumber\\
%&\qquad \mathrm{C5}:  \mathbf{Q}_n\succeq\mathbf{0},~~\forall n \in \mathcal{N},
% \nonumber
%\end{align}
%%where  $\mathbf{c}=\{c_n\}_{n=1}^N$ and $\mathbf{w}=\{\mathbf{w}_n\}_{n=1}^N$.
%%as illustrated in subsection \ref{edge_cloud_latency} and \ref{cloud_time-process}
%%where $0\leq \alpha<1$ is the predefined fraction for a specified scenario.
In problem \eqref{P1}, $\mathrm{C2}$ represents the central/backhaul processing  latency constraint, indicating that the central cloud is selected, i.e., the backhaul is allowed to be used for task offloading, only when the setting parameters can make sure that the central/backhaul  processing latency is lower than certain percentage, e.g., $\alpha$, of edge computing latency. Considering the scarce  backhaul resources, this constraint is reasonable in practice and {of great benefit} to guarantee the high-speed backhaul transmission, avoid the abuse of backhaul and alleviate the backhaul congestion.
Here, $0< \alpha<1$ is a predefined fraction for a specified scenario depending on the central cloud and backhaul restriction.
For the special case of $\alpha=0$, central cloud becomes unavailable as indicated in $\mathrm{C2}$ and $c_n=1$ for $n\in\mathcal{N}$,  then problem \eqref{P1} reduces to resource allocation problem in traditional MEC networks, which has been studied from different perspectives in the literatures such as  \cite{chenxu_2016,Changsheng_you_mar_2017,A_o_simeone,Cheng_wang_aug_2017,X_Lyu_2017,X_Sun_2017,C_You_2016,Xiaoyan_hu_2018,Jiexu_WPT_MEC2018,J_S.Bi18Computation}.
$\mathrm{C3}$  is the latency constraint for edge processing, such that the sum of the access transmission latency and the edge computing latency  should not exceed a given threshold $T_\mathrm{th}$. %\footnote{In our considered scenario, we assume that the UEs' tasks have already been synchronized. In fact, our work can be easily extended into the cases considering the  latency of synchronizing UEs' tasks. For the case with deterministic task arrival model, the edge processing latency constraints C3  should be changed into  $T_n^\mathrm{a}(\mathbf{p}^{\mathrm{u}},\mathbf{w}_n)+T{_n^{\mathrm{syn}}}+c_n T_n^{\mathrm{edge}} \leq T_\mathrm{th},n\in\mathcal{N}$, where $T{_n^{\mathrm{syn}}}$ is the synchronization latency for UE $n$. For the case with random task arrive model, we can introduce a maximum synchronization latency threshold, denoted as $T_{\mathrm{syn}}$. Then constraints C3 can be changed into  $T_n^\mathrm{a}(\mathbf{p}^{\mathrm{u}},\mathbf{w}_n)+c_n T_n^{\mathrm{edge}} \leq T_\mathrm{th}-T_{\mathrm{syn}}, n\in\mathcal{N}$. In this way, we can also leverage the proposed algorithm in our manuscript to solve the corresponding formulated problem.}
Note that $T_n^{\mathrm{edge}}$ expressed in \eqref{time_edge} increases with the task size $I_n$,  and thus if edge cloud cannot meet its latency constraint in $\mathrm{C3}$ when encounters large tasks, e.g., $T_n^{\mathrm{edge}}>T_\mathrm{th}$, central cloud will be the only option to be utilized, which further indicates the complementary relationship between edge and central cloud computing \cite{ETSI_2018_white_paper}. $\mathrm{C4}$ and $\mathrm{C5}$ guarantee that transmit power values are non-negative.

{In our considered scenario, we assume that  UEs' tasks have already been synchronized. In fact, our work can be easily extended into the cases considering the  latency of synchronizing UEs' tasks. For the case with deterministic task arrival model \cite{J_Y.Mao17Mobile}, the edge processing latency constraints $\mathrm{C3}$  should be changed into  $T{_n^{\mathrm{syn}}}+T_n^\mathrm{a}(\mathbf{p}^{\mathrm{u}},\mathbf{w}_n)+c_n T_n^{\mathrm{edge}} \leq T_\mathrm{th},n\in\mathcal{N}$, where $T{_n^{\mathrm{syn}}}$ is the synchronization latency of UE $n$. For the case with random task arrival model \cite{J_Y.Mao17Mobile}, we can introduce a maximum synchronization latency threshold, denoted as $T_{\mathrm{syn}}$. Then constraints $\mathrm{C3}$ can be changed into  $T_n^\mathrm{a}(\mathbf{p}^{\mathrm{u}},\mathbf{w}_n)+c_n T_n^{\mathrm{edge}} \leq T_\mathrm{th}-T_{\mathrm{syn}}, n\in\mathcal{N}$. In this way, we can also leverage the algorithms proposed in section \ref{algorithm_design} to solve the corresponding formulated problems.}

\section{Algorithm Design}\label{algorithm_design}
The considered problem \eqref{P1} is a mixed-integer and non-convex optimization problem because of the integer cloud selection indicator $\mathbf{c}$, and the non-convex objective function and constraints $\mathrm{C2}$, $\mathrm{C3}$, which is {NP-hard} in general and its optimal solution is difficult to achieve.  To be tractable, we first need to determine whether edge or central cloud computing will be employed, and  then we can optimize the transmit powers, receive beamformers and covariance matrices. Hence, a tractable  decomposition approach can be developed to solve  \eqref{P1} in an iterative manner considering the fact that $\mathbf{c}$ and $\left\{\mathbf{p}^{\mathrm{u}},\mathbf{w}, \mathbf{Q}\right\}$ are coupled in the objective function and constraints $\mathrm{C2}$, $\mathrm{C3}$  of problem \eqref{P1}.
\subsection{Edge or Central Cloud Computing}\label{cloud_selection}
As mentioned in section \ref{Problem_Formulation}, when the $n$-th edge cloud's computing time  $T_n^{\mathrm{edge}}$ is greater than the maximum allowable time $T_\mathrm{th}$, the use of edge cloud is infeasible and central cloud computing has to be utilized, i.e., $c_n=0$. Next, we optimize the cloud selection indicator  $\mathbf{c}$  for the case of $T_n^{\mathrm{edge}}<T_\mathrm{th}$.
To properly deal with the {integer optimization} caused by  $c_n$, we first relax $c_n\in \left\{0,1\right\}$ as  $\widehat{c}_n \in \left[0,1\right]$, and denote $\mathbf{\widehat{c}}=\{\widehat{c}_n\}_{n=1}^N$ as the set of the relaxed cloud selection variable $\widehat{c}_n$. Then problem \eqref{P1} with given feasible $\left\{\mathbf{p}^{\mathrm{u}},\mathbf{w}, \mathbf{Q}\right\}$ can be decomposed into the following relaxed version
\begin{align}\label{P1_relax}
&\mathop {\min }\limits_{ \mathbf{\widehat{c}} }~\sum\limits_{n{\rm{ = }}1}^N \left(\widehat{c}_n E_n^\mathrm{edge}+\left(1-\widehat{c}_n\right)E_n^\mathrm{central}\right)\\
&~\mathrm{s.t.} ~~~  \mathrm{\widehat{C}1}:  \widehat{c}_n \in \left[0,1\right],~~\forall n \in \mathcal{N}, \nonumber\\%~~~\mathrm{C2},~~~\mathrm{C3}.    \nonumber
&~~\qquad \mathrm{\widehat{C}2}: \left(1-\widehat{c}_n\right) T_n^{\mathrm{central}}(\mathbf{Q}) \leq  \alpha T_n^{\mathrm{edge}},~~\forall n \in \mathcal{N},\nonumber\\
&~~\qquad  \mathrm{\widehat{C}3}: T_n^\mathrm{a}(\mathbf{p}^{\mathrm{u}},\mathbf{w}_n)+\widehat{c}_n T_n^{\mathrm{edge}} \leq T_\mathrm{th},~~\forall n \in \mathcal{N}.  \nonumber
\end{align}
Problem \eqref{P1_relax} is a one-dimensional linear programming, and its solution can be given in the following two cases:
\begin{itemize}
\item Case 1: Without loss of generality, if the energy consumption of edge computing is lower than that of central processing for UE $n$'s task, i.e, $E_n^\mathrm{edge} \leq E_n^\mathrm{central}$, the objective function in \eqref{P1_relax} is a decreasing function of $\widehat{c}_n$. Therefore, the optimal $\widehat{c}_n^*$ is the maximum value that satisfies  $\mathrm{\widehat{C}1-\widehat{C}3}$, i.e.,
\begin{equation}
\widehat{c}_n^*=\left[\min\left\{\frac{T_\mathrm{th}-T_n^\mathrm{a}(\mathbf{p}^{\mathrm{u}},\mathbf{w}_n)}{T_n^{\mathrm{edge}}},1\right\}\right]^+.
\end{equation}

\item Case 2: if $E_n^\mathrm{edge} > E_n^\mathrm{central}$, the objective function in \eqref{P1_relax} is an increasing function of $\widehat{c}_n$, and the optimal $\widehat{c}_n^*$ is the minimum value that satisfies $\mathrm{\widehat{C}1-\widehat{C}3}$, i.e.,
\begin{equation}
\widehat{c}_n^*=\left[1-\frac{\alpha T_n^{\mathrm{edge}}}{T_n^{\mathrm{central}}(\mathbf{Q}) }\right]^+.%+{T_n^{\mathrm{pre}}}
\end{equation}
\end{itemize}
It is seen that the relaxed edge/central cloud computing decision $\mathbf{\widehat{c}}^*$ is reliant on the optimal $\left\{\mathbf{p}^{\mathrm{u}},\mathbf{w}, \mathbf{Q}\right\}$ of problem \eqref{P1}. In the following two subsections, we will focus on obtaining the optimal $\{\mathbf{p}^{{\mathrm{u}}*},\mathbf{w}^*\}$ and $\mathbf{Q}^*$, respectively, based on a given cloud selection decision $\mathbf{\widehat{c}}$. %, {which is obtained as the integer approximation of $\widehat{\mathbf{c}}$}.

\subsection{UEs' Transmit Powers and SBSs' Receive Beamformers}
For fixed cloud selection decision $\mathbf{\widehat{c}}$, the optimal $\left\{\mathbf{p}^{\mathrm{u}*},\mathbf{w}^*\right\}$ can be obtained by solving the subproblem of \eqref{P1} as follows:
\begin{align}\label{Access_power_beam}
&\mathop {\min }\limits_{\mathbf{p}^{\mathrm{u}},\mathbf{w}}~\sum\limits_{n{\rm{ = }}1}^N {p_n^{\mathrm{u}} T_n^\mathrm{a}}(\mathbf{p}^{\mathrm{u}},\mathbf{w}_n) \\
&~\mathrm{s.t.} ~~~\mathrm{\widehat{C}3},~~~\mathrm{C4},    \nonumber
\end{align} % 0\leq p_n^{\mathrm{u}} \leq P_{\max}^{\rm u},\forall n \in \mathcal{N},
where $\widehat{\mathrm{C}}3$ and $\mathrm{C}4$ are the corresponding constraints expressed in problem \eqref{P1_relax} and \eqref{P1}, respectively.
The subproblem \eqref{Access_power_beam} is non-convex (over $\mathbf{p}^{\mathrm{u}}$) and its objective function is
the weighted sum-of-ratios, which is challenging to solve. We first examine the interplay between UEs' transmit power vector
 $\mathbf{p}^{\mathrm{u}}$ and  SBSs' receive beamformers in $\mathbf{w}$.
\begin{lemma}\label{lemma1}
For fixed $\mathbf{p}^{\mathrm{u}}$, the optimal $\mathbf{w}_n^*$ of problem \eqref{Access_power_beam} is given by
\begin{align}\label{w_beam}
\mathbf{w}_n^*={\rm eigvec}\left\{\max\left\{{\rm eig}\{\left(\mathbf{\Omega}_{-n}\right)^{-1} \mathbf{\Omega}_n \}\right\}\right\},
\end{align}
where $\mathbf{\Omega}_{-n}=\sigma_n^2\mathbf{I}_L+\sum_{i=1, i\neq n}^{N}p_i^\mathrm{u}{\mathbf{h}_{i,n}^{\rm a}}
({\mathbf{h}_{i,n}^{\rm a}})^H$ and $\mathbf{\Omega}_n =p_n^{\mathrm{u}} {\mathbf{h}_{n,n}^{\rm a}}({\mathbf{h}_{n,n}^{\rm a}})^H$.
\end{lemma}

\begin{proof}
See Appendix A.
\end{proof}

With the help of auxiliary variables $\mathbf{t}=\{t_n\}_{n=1}^N$, problem \eqref{Access_power_beam} over the UEs' transmit power vector $\mathbf{p}^{\mathrm{u}}$ for fixed $\mathbf{w}$ can be equivalently transformed as
\begin{align}\label{P1_max_sub1}
&\mathop {\min }\limits_{\mathbf{p}^{\mathrm{u}},\mathbf{t}}~ \sum\limits_{n{\rm{ = }}1}^N I_n t_n\\
&~\mathrm{s.t.} ~~\mathrm{\widetilde{C}1}: \frac{p_n^{\mathrm{u}}}{R_n^{\mathrm{a}}(\mathbf{p}^{\mathrm{u}},\mathbf{w}_n)}\leq t_n,\ \forall n \in \mathcal{N},\nonumber\\
&~\qquad \mathrm{\widetilde{C}2}: \gamma_n^{\mathrm{a}}(\mathbf{p}^{\mathrm{u}},\mathbf{w}_n)\geq \tau,\ \forall n \in \mathcal{N}, \nonumber\\
&~\qquad \mathrm{\widetilde{C}3}: 0\leq p_n^{\mathrm{u}} \leq P_{\max}^{\rm u}, \ \forall n \in \mathcal{N},   \nonumber
\end{align}
where $\tau=2^{\frac{I_n}{B^{\mathrm{a}}\left(T_\mathrm{th}-\widehat{c}_n T_n^{\mathrm{edge}}\right)}}-1$.

\begin{lemma}\label{lemma2}
The optimal solution $\left(\mathbf{p}^{\mathrm{u}*},\mathbf{t}^*\right)$ of problem \eqref{P1_max_sub1}  satisfies the  Karush-Kuhn-Tucker (KKT) conditions of the following $N$ $(n\in\mathcal{N})$ subproblems
\begin{align}\label{P1_max_sub1_sub2}
&\mathop {\min }\limits_{p_n^{\mathrm{u}}} ~\left(\lambda_n+M_n \right)p_n^{\mathrm{u}}- \lambda_n t_n R_n^{\mathrm{a}}(\mathbf{p}^{\mathrm{u}},\mathbf{w}_n)\\
%&~\mathrm{s.t.} ~~~\mathrm{\widetilde{C}2}, ~~~\mathrm{\widetilde{C}3},   \nonumber
&~\mathrm{s.t.} ~~~\mathrm{\widetilde{C}2}:  \gamma_n^{\mathrm{a}}(\mathbf{p}^{\mathrm{u}},\mathbf{w}_n)\geq \tau,   \nonumber\\
&~~~~~~~~\mathrm{\widetilde{C}3}: 0\leq p_n^{\mathrm{u}} \leq P_{\max}^{\rm u},   \nonumber
\end{align}
with
\begin{align}\label{M_n_num}
\quad M_n=\sum_{j=1, j\neq n}^N\lambda_j t_j \frac{B_\mathrm{a}}{\ln2}\frac{ \left(\gamma_j^{\mathrm{a}} \right)^2 |\mathbf{w}_j^H
{\mathbf{h}_{n,j}^{\rm a}}|^2}{p_j^{\mathrm{u}} |\mathbf{w}_j^H{\mathbf{h}_{j,j}^{\rm a}}|^2 \left(1+\gamma_j^{\mathrm{a}}\right)}+
\sum_{j=1, j\neq n}^N\mu_j \frac{ \left(\gamma_j^{\mathrm{a}} \right)^2 |\mathbf{w}_j^H {\mathbf{h}_{n,j}^{\rm a}}|^2}{p_j^{\mathrm{u}}
|\mathbf{w}_j^H {\mathbf{h}_{j,j}^{\rm a}}|^2 },
\end{align}
%\begin{align}\label{M_n_num}
%M_n=\sum_{j=1, j\neq n}^N\lambda_j t_j \frac{B_\mathrm{a}}{\ln2}\frac{ \left(\gamma_j^{\mathrm{a}} \right)^2 |\mathbf{w}_j^H
%{\mathbf{h}_{n,j}^{\rm a}}|^2}{p_j^{\mathrm{u}} |\mathbf{w}_j^H{\mathbf{h}_{j,j}^{\rm a}}|^2 \left(1+\gamma_j^{\mathrm{a}}\right)}+
%\sum_{j=1, j\neq n}^N\mu_j \frac{ \left(\gamma_j^{\mathrm{a}} \right)^2 |\mathbf{w}_j^H {\mathbf{h}_{n,j}^{\rm a}}|^2}{p_j^{\mathrm{u}}
%|\mathbf{w}_j^H {\mathbf{h}_{j,j}^{\rm a}}|^2 },
%\end{align}
where $\{\lambda_n\}_{n=1}^N$ and $\{{\mu_n}\}_{{n}=1}^N$ are Lagrange multipliers associated with the constraints $\mathrm{\widetilde{C}1}$ and $\mathrm{\widetilde{C}2}$ of problem \eqref{P1_max_sub1}, respectively, and $M_n=-\sum_{j\neq n}^N\lambda_j t_j\frac{\partial R_j^{\mathrm{a}} }{\partial p_n^{\mathrm{u}}}-\sum_{j\neq n}^N\mu_j\frac{\partial \gamma_j^\mathrm{a}}{\partial p_n^{\mathrm{u}}}$. For optimal $\left(\mathbf{p}^{\mathrm{u}*},\mathbf{w}^*\right)$, $\lambda_n$ and $t_n$ are respectively calculated as
%$\lambda_n=\frac{I_n}{R_n^{\mathrm{a}}\left(\mathbf{p}^{\mathrm{u}*},\mathbf{w}_n^*\right)}$ and $t_n=\frac{p_n^{\mathrm{u}*}}{R_n^{\mathrm{a}}\left(\mathbf{p}^{\mathrm{u}*},\mathbf{w}_n^*\right)}$, respectively.
\begin{align}
  \lambda_n=\frac{I_n}{R_n^{\mathrm{a}}\left(\mathbf{p}^{\mathrm{u}*},\mathbf{w}_n^*\right)},\label{eq:lambda_n}\\
  t_n=\frac{p_n^{\mathrm{u}*}}{R_n^{\mathrm{a}}\left(\mathbf{p}^{\mathrm{u}*},\mathbf{w}_n^*\right)}.\label{eq:t_n}
\end{align}
%Then we can obtain the optimal $\mathbf{p}^{\mathrm{u}*}$  by solving the $N$ sub-problems in \eqref{P1_max_sub1_sub2}.
\end{lemma}

\begin{proof}
See Appendix B.
\end{proof}

Given  $\lambda_n$ and $t_n$, the subproblem \eqref{P1_max_sub1_sub2} is convex w.r.t. $p_n^{\mathrm{u}}$. Therefore, we have the following theorem.

\begin{theorem}\label{theorem_1}
The solution of subproblem \eqref{P1_max_sub1_sub2} is given by
%\begin{align}\label{theo_1_power}
%& p_n^{\mathrm{u}*}=\left\{ \begin{array}{l}
%  \frac{\tau}{\Lambda_n}  ,~\mathrm{if}~G_n<\frac{\tau}{\Lambda_n},\\
%G_n , ~\mathrm{if}~\frac{\tau}{\Lambda_n} \leq G_n \leq  P_{\max}^{\rm u},\\
%P_{\max}^{\rm u},~\mathrm{if}~G_n > P_{\max}^{\rm u},
%\end{array} \right.  \\
%& \mu_n^*=\left\{ \begin{array}{l}
% \frac{\lambda_n+M_n}{\Lambda_n}-\frac{B \lambda_n t_n}{\left(\tau+1\right)\ln2}  ,~\mathrm{if}~G_n<\frac{\tau}{\Lambda_n},\\
%0 , ~\mathrm{otherwise}
%\end{array} \right.
%& \nu_n^*=\left\{ \begin{array}{l}
%0 , ~\mathrm{if}~G_n \leq  P_{\max}^{\rm u},
%\\
%\frac{B}{\ln2} \frac{\lambda_n t_n }{P_{\max}^{\rm u}+1/\Lambda_n}  -\lambda_n-M_n,~\mathrm{otherwise},
%\end{array} \right.
%\end{align}
\begin{align}
&&&p_n^{\mathrm{u}*}=\left\{
\begin{aligned}
&\frac{\tau}{\Lambda_n}, ~~~~\mathrm{if}~G_n<\frac{\tau}{\Lambda_n}, \\
&G_n, ~~~~\mathrm{if}~\frac{\tau}{\Lambda_n} \leq G_n \leq  P_{\max}^{\rm u}, \\
&P_{\max}^{\rm u}, ~~\mathrm{if}~G_n > P_{\max}^{\rm u},
\end{aligned}\right. \label{theo_1_power}\\
&&&\mu_n^*=\left\{
\begin{aligned}
&\frac{\lambda_n+M_n}{\Lambda_n}-\frac{B_\mathrm{a}}{\ln2} \frac{\lambda_n t_n}{\tau+1}, \quad\mathrm{if}~G_n<\frac{\tau}{\Lambda_n}, \\
&0, \quad\mathrm{otherwise}, \\
\end{aligned}\right. \label{theo_1_mu}\\
&&&\nu_n^*=\left\{
\begin{aligned}
&0, \quad\mathrm{if}~G_n \leq  P_{\max}^{\rm u}, \\
&\frac{B_\mathrm{a}}{\ln2} \frac{\lambda_n t_n }{P_{\max}^{\rm u}+1/\Lambda_n}-\lambda_n-M_n, ~\mathrm{otherwise}, \\
\end{aligned}\right. \label{theo_1_nu}
\end{align}
where we define $\Lambda_n\triangleq\frac{|\mathbf{w}_n^H {\mathbf{h}_{n,n}^{\rm a}}|^2}{\sum_{i=1,i\neq n}^{N}
p_i^{\mathrm{u}}|\mathbf{w}_n^H {\mathbf{h}_{i,n}^{\rm a}}|^2+|\mathbf{w}_n^H\mathbf{n}_n|^2}$, $G_n\triangleq\frac{B_\mathrm{a}}{\ln2}
\frac{\lambda_n t_n }{\lambda_n+M_n}-\frac{1}{\Lambda_n}$, and $\mu_n^*$ and $\nu_n^*$ are respectively  the optimal Lagrange multipliers associated with the constraints $\widetilde{\mathrm{C}}2$ and $\widetilde{\mathrm{C}}3$ of problem \eqref{P1_max_sub1_sub2}.
\end{theorem}

\begin{proof}
See Appendix C.
\end{proof}

In light of the results in \textbf{Lemma}~\ref{lemma1}, \textbf{Lemma}~\ref{lemma2} and \textbf{Theorem}~\ref{theorem_1}, we provide an iterative approach to effectively solve problem \eqref{Access_power_beam}, which is shown in \textbf{Algorithm}~\ref{algorithmic1}.

\begin{algorithm}[!htp]
{\small{
\caption{ Solution of Problem \eqref{Access_power_beam}}\label{algorithmic1}
\begin{algorithmic}[1]
\STATE \textbf{Initialize} $p_n^{\mathrm{u}}=P_{\max}^{\rm u}$,  $\forall n$. Set $\mathbf{w}_n$ based on \textbf{Lemma}~\ref{lemma1}.
\STATE \textbf{Repeat}
\STATE \qquad a) Given $\mathbf{w}$, Loop: \\
               \qquad \quad i): Compute $M_n$, $\lambda_n$ and $t_n$ based on \textbf{Lemma}~\ref{lemma2}. \\
               \qquad \quad ii): Update $p_n^{\mathrm{u}}$ and $\mu_n$ based on \textbf{Theorem}~\ref{theorem_1}. \\
                \qquad \quad Until convergence.
\STATE \qquad b) Update $\mathbf{w}$ based on \textbf{Lemma}~\ref{lemma1}.
\STATE \;\;\; Until convergence, and obtain the optimal $\{\mathbf{p}^{\mathrm{u}*},\mathbf{w}^*\}$.
\end{algorithmic}
}}
\end{algorithm}

The convergence of \textbf{Algorithm}~\ref{algorithmic1} can be guaranteed since the objective function of problem~\eqref{Access_power_beam} decreases with the iteration index (in step 3 and step 4 of \textbf{Algorithm}~\ref{algorithmic1}), which is indicated from optimizing $\mathbf{p}^\mathrm{u}$ and $\mathbf{w}$ in each iteration as shown in \textbf{Lemma}~\ref{lemma1} and \textbf{Lemma}~\ref{lemma2}, respectively.

\subsection{SBSs' Transmit Covariance Matrices}
For fixed cloud selection decision $\mathbf{\widehat{c}}$, the optimal $\mathbf{Q}^*$ can be obtained by solving the following subproblem:
\begin{align}\label{P1_2_Q}
&\mathop {\min }\limits_{\mathbf{Q} }~y\left(\mathbf{Q}\right)=\sum\limits_{n{\rm{ = }}1}^N {\left( {1{\rm{ - }}{ \widehat{c}_n } } \right)
\mathrm{tr}\left(\mathbf{Q}_n\right) T_n^{\mathrm{central}}(\mathbf{Q}) } \\
&~\mathrm{s.t.}~~\mathrm{\widehat{C}2}: R_n^{\mathrm{b}}(\mathbf{Q}) \geq \left(1-\widehat{c}_n\right)\frac{I_n}{\alpha T_n^{\mathrm{edge}}},\ \forall n \in \mathcal{N}, ~~\mathrm{C5}, \nonumber
\end{align}
where $\mathrm{\widehat{C}}2$ and $\mathrm{C}5$ are the corresponding  constraints expressed in problem \eqref{P1_relax} and \eqref{P1}, respectively, and  $\mathrm{\widehat{C}}2$ is re-expressed in an equivalent form here.
Problem \eqref{P1_2_Q} is non-convex due to the non-convexity of the objective function and constraint $\mathrm{\widehat{C}}2$, which cannot be solved directly. Thus, we resort to a successive pseudoconvex approach, which has many advantages such as fast convergence and  parallel computation~\cite{J_Y.Yang17A_unified}. % to solve it

First, let $\mathbf{Q}^l$ denote the $\mathbf{Q}$ value in the $l$-th iteration. Thus the non-convex item $\mathrm{tr}\left(\mathbf{Q}_n\right) T_n^{\mathrm{central}}(\mathbf{Q})$ for each $n\in\mathcal{N}$ in the objective function can be approximated as a pseudoconvex function at $\mathbf{Q}^l$, which is written as
\begin{align}\label{pseudo_convex_obj}
\widehat{y}_n(\mathbf{Q}_n;\mathbf{Q}^l)\triangleq
 \frac{I_n\mathrm{tr}(\mathbf{Q}_n)}{R_n^{\rm b}(\mathbf{Q}_n;\mathbf{Q}^l)}+\chi_n(\mathbf{Q}_n),
\end{align}
where $\chi_n(\mathbf{Q}_n)=\sum_{j\neq n} I_j\mathrm{tr}(\mathbf{Q}_j^l)\Big\langle(\mathbf{Q}_n-\mathbf{Q}_n^l),
  \nabla_{\mathbf{Q}_n^\dagger}\frac{{1-\widehat{c}_j}}{R_j^{\rm b}(\mathbf{Q}^l)}\Big\rangle$ is a function obtained by linearizing the non-convex function $\sum_{j\neq n}^N\mathrm{tr}\left(\mathbf{Q}_j\right) T_j^{\mathrm{central}}(\mathbf{Q})$ in $\mathbf{Q}_n$ at the point $\mathbf{Q}^l$. % and $\nabla_{\mathbf{Q}_j^\dagger}\frac{{1-c_j}}{R_j^{\rm b}(\mathbf{Q}^l)}$ is the Jacobian matrix of $\frac{{1-c_j}}{R_j^{\rm b}(\mathbf{Q}^l)}$ w.r.t. $\mathbf{Q}_j^\dagger$.
Based on \eqref{pseudo_convex_obj},  we can approximate the objective function $y\left(\mathbf{Q}\right)$ of problem \eqref{P1_2_Q} at $\mathbf{Q}^l$ as
\begin{align}\label{obj_approx_Qq}
\widehat{y}(\mathbf{Q};\mathbf{Q}^l)=\sum\limits_{n{\rm{ = }}1}^N \left( {1{\rm{ - }}{\widehat{c}_n}} \right)\widehat{y}_n(\mathbf{Q}_n;\mathbf{Q}^l).
\end{align}
It is easily seen that $\widehat{y}(\mathbf{Q};\mathbf{Q}^l)$ is pseudoconvex and has the same gradient with $y\left(\mathbf{Q}\right)$ at $\mathbf{Q}=\mathbf{Q}^l$ \cite{J_Y.Yang17A_unified}.

Then, by equivalently rewriting the non-concave function $R_n^{\mathrm{b}}(\mathbf{Q})$ in constraint $\mathrm{\widehat{C}2}$ as a difference of two concave functions as expressed in \eqref{Q_constra_1a}  according to its definition in \eqref{backhaul_rate}, and leveraging the first-order Taylor expansion at $\mathbf{Q}^l$ for the second concave function denoted as  ${R}_n^{\mathrm{b2}}(\mathbf{Q})=B^\mathrm{b} \log_2\det\left(\sigma^2\mathbf{I}+ \sum_{i\neq n}^N\mathbf{H}_i^\mathrm{b} \mathbf{Q}_i \left(\mathbf{H}_i^\mathrm{b}\right)^H\right)$, $R_n^{\mathrm{b}}(\mathbf{Q})$ can be approximated as
\begin{subeqnarray} \label{Q_constra_1}
\hspace{-6mm}R_n^{\mathrm{b}}(\mathbf{Q})&&\hspace{-6mm}=B^\mathrm{b} \log_2\det\left(\sigma^2\mathbf{I}+\mathbf{\Xi}(\mathbf{Q})\right)-{R}_n^{\mathrm{b2}}(\mathbf{Q}) \slabel{Q_constra_1a}\\
&&\hspace{-6mm}\geqslant
B^\mathrm{b} \log_2\det\left(\sigma^2\mathbf{I}+\mathbf{\Xi}(\mathbf{Q})\right)-{R}_n^{\mathrm{b2}}(\mathbf{Q}^l)-\nonumber\\
&&\hspace{-6mm} \sum_{j\neq n}^N\Big\langle(\mathbf{Q}_j-\mathbf{Q}_j^l),
  \nabla_{\mathbf{Q}_j^\dagger}R_n^{\mathrm{b2}}(\mathbf{Q}^l)\Big\rangle\triangleq\bar{R}_n^{\mathrm{b}}(\mathbf{Q}), \slabel{Q_constra_1b}
\end{subeqnarray}
%Then, by leveraging the first-order Taylor expansion at $\mathbf{Q}^l$, the left-hand-side of the non-convex constraint $\mathrm{C3}$ can be approximated as
% \begin{align}\label{Q_constra_1}
% R_n^{\mathrm{b}}(\mathbf{Q})&=B^\mathrm{b} \log_2\det\left(\sigma^2\mathbf{I}+\mathbf{\Xi}(\mathbf{Q})\right)-{R}_n^{\mathrm{b2}}(\mathbf{Q})\nonumber\\
% &\approx
%B^\mathrm{b} \log_2\det\left(\sigma^2\mathbf{I}+\mathbf{\Xi}(\mathbf{Q})\right)-{R}_n^{\mathrm{b2}}(\mathbf{Q}^l)-\nonumber\\
% & \sum_{j\neq n}^N\Big\langle(\mathbf{Q}_j-\mathbf{Q}_j^l),
%  \nabla_{\mathbf{Q}_j^\dagger}R_n^{\mathrm{b2}}(\mathbf{Q}^l)\Big\rangle\triangleq\bar{R}_n^{\mathrm{b}}(\mathbf{Q}),
%   \end{align}
where $\mathbf{\Xi}\left(\mathbf{Q}\right)=\sum_{i=1}^N\mathbf{H}_i^\mathrm{b} \mathbf{Q}_i \left(\mathbf{H}_i^\mathrm{b}\right)^H$.  Here, $\bar{R}_n^{\mathrm{b}}(\mathbf{Q})$ expressed in \eqref{Q_constra_1b} is a concave function over $\mathbf{Q}$.

Therefore, at $\mathbf{Q}^l$, the original problem \eqref{P1_2_Q} can be approximately transformed as
\begin{align}\label{Q_l_app_1}
&\mathop {\min }\limits_{\mathbf{Q} }~\widehat{y}(\mathbf{Q};\mathbf{Q}^l) \\
&~\mathrm{s.t.} ~~\mathrm{\overline{C}2}: \bar{R}_n^{\mathrm{b}}(\mathbf{Q})\geq \left(1-\widehat{c}_n\right)\frac{I_n}{\alpha T_n^{\mathrm{edge}}},\ \forall n \in \mathcal{N},  ~~\mathrm{C5}. \nonumber
\end{align}
The objective function of problem \eqref{Q_l_app_1} is a sum  of $N$ pseudoconvex functions each containing  a  fractional function and  a linear function. In addition, all the constraints in problem \eqref{Q_l_app_1} are convex.
Hence, by leveraging the Dinkelbach-like algorithm~\cite{J_A.Zappone15Energy} and introducing a set of auxiliary variables for the $N$
fractional functions in the objective function, problem \eqref{Q_l_app_1} can be transformed into a solvable convex  optimization problem,
which can be effectively solved by CVX \cite{B_Grant08CVX} and owns provable convergence~\cite{J_Y.Yang17A_unified}.  %$\chi=\{\chi_n\}_{n=1}^N$
Let $\mathbb{B}\mathbf{Q}^l$ represent the optimal solution of problem \eqref{Q_l_app_1} at the $l$-th iteration,
%\begin{lemma}\label{pro:2}
% A point $\mathbf{Q}^l$ is a KKT point of the original problem (P4) if and only if $\mathbf{Q}^l=\mathbb{B}\mathbf{Q}^l$. If $\mathbf{Q}^l$ is not a KKT point, then $\mathbb{B}\mathbf{Q}^l-\mathbf{Q}^l$ is a descent direction of $y(\mathbf{Q})$ such that $(\mathbb{B}\mathbf{Q}^l-\mathbf{Q}^l)\nabla_{\mathbf{Q}^\dag} y(\mathbf{Q}^l)<0$.
%\end{lemma}
and thus the value of $\mathbf{Q}$  in the next $(l+1)$-th iteration can be updated as
\begin{align}\label{eq:Q_update}
  \mathbf{Q}^{l+1}= \mathbf{Q}^{l}+\varsigma(l)(\mathbb{B}\mathbf{Q}^l-\mathbf{Q}^l),
\end{align}
where $\varsigma(l)$ is the step size at the $l$-th iteration and can be obtained through the successive line search, and $\mathbb{B}\mathbf{Q}^l-\mathbf{Q}^l$ is the descent direction of $y\left(\mathbf{Q}\right)$. Thus, the solution of problem \eqref{P1_2_Q} can be iteratively obtained.

Based on the aforementioned analysis of optimizing the variables $\{\mathbf{p}^{{\mathrm{u}}*},\mathbf{w}^*,\mathbf{Q}^*\}$, \textbf{Algorithm}~\ref{algorithmic2} is proposed to solve the original problem \eqref{P1}.

\begin{algorithm}[!htp]
{\small{
\caption{ Solution of Problem \eqref{P1}}\label{algorithmic2}
\begin{algorithmic}[1]
\STATE \textbf{Initialize} $p_n^{\mathrm{u}}=P_{\max}^{\rm u}$,  $\forall n$. Set $\mathbf{w}_n$ based on \textbf{Lemma}~\ref{lemma1}.\\
    \quad Based on the constraint $\mathrm{\widehat{C}3}$ of problem \eqref{P1_relax}, set {$\widehat{c}_n=\left[\min\left\{\frac{T_\mathrm{th}-T_n^\mathrm{a}(\mathbf{p}^{\mathrm{u}},\mathbf{w}_n)}
    {T_n^{\mathrm{edge}}},1-\delta\right\}\right]^+$, where $\delta\in(0,0.5)$ is a tolerant value to avoid the selection of solely edge clouds or central cloud at the initial point}. Then, based on the constraint $\mathrm{\widehat{C}2}$ of problem \eqref{P1_relax}, $\mathbf{Q}$ is set to meet $ T_n^{\mathrm{central}}(\mathbf{Q}) = \frac{\alpha T_n^{\mathrm{edge}}}{1-\widehat{c}_n}$ through the use of ZF precoding with equal power allocation at each SBS.
\STATE \textbf{Repeat}
\STATE \qquad a) Given ${\{\widehat{c}_n\}_{n=1}^N}$: \\
               \qquad \quad i): Update $\{\mathbf{p}^{\mathrm{u}},\mathbf{w}\}$ based on \textbf{Algorithm}~\ref{algorithmic1}. \\
               \qquad \quad ii): Loop: \\
                \qquad \qquad~~~ii-1): Solve problem \eqref{Q_l_app_1} via Dinkelbach-like algorithm~\cite{J_A.Zappone15Energy}. \\
                \qquad \qquad~~~ii-2): Update $\mathbf{Q}^l$ based on \eqref{eq:Q_update}.  \\
                 \qquad \qquad~~  Until convergence, and obtain the updated $\mathbf{Q}$.  \\
\STATE \qquad b) Update ${\{\widehat{c}_n\}_{n=1}^N}$ according to subsection III-A.
\STATE \;\;\; Until convergence, and obtain solution $\{\mathbf{c}^*,\mathbf{p}^{\mathrm{u}*},\mathbf{w}^*,\mathbf{Q}^*\}$,
in which $\mathbf{c}^*$ is obtained by rounding the cloud selection solution of problem \eqref{P1_relax}, i.e., $\mathbf{\widehat{c}}$, and $\mathbf{p}^{\mathrm{u}*},\mathbf{w}^*,\mathbf{Q}^*$ are obtained based on $\mathbf{c}^*$.
\end{algorithmic}
}}
\end{algorithm}

\subsection{Convergence and Complexity}
The convergence of \textbf{Algorithm}~\ref{algorithmic2} is easy to prove in light of the guaranteed convergence of \textbf{Algorithm}~\ref{algorithmic1}, the Dinkelbach-like algorithm used to solve problem~\eqref{Q_l_app_1} \cite{J_A.Zappone15Energy}, and the update process of the cloud selection $\mathbf{\widehat{c}}$ illustrated in Section~\ref{cloud_selection}. Note that the objective function of problem~\eqref{P1_relax} is a decreasing function of the iteration index (in step 3 and step 4 of \textbf{Algorithm}~\ref{algorithmic2}),  which ensures the convergence of \textbf{Algorithm}~\ref{algorithmic2}.

The proposed  \textbf{Algorithm}~\ref{algorithmic2} enjoys an acceptable complexity as well as an easy implementation. In each iteration,
the majority of computational complexity lies in solving subproblem \eqref{Access_power_beam} for obtaining the optimal
$(\mathbf{p}^{\mathrm{u}*},{\mathbf{w}}^*)$ and
the approximate subproblem \eqref{Q_l_app_1} for obtaining the optimal $\mathbf{Q}^*$ with a given $\mathbf{\widehat{c}}$. In the proposed algorithm, problem \eqref{Access_power_beam} can be equivalently
transformed into $N$ independent subproblems \eqref{P1_max_sub1_sub2} and  thus  can be easily solved in a parallel way. Moreover,
the optimal solution of each subproblem has closed-form expressions as indicated in \textbf{Theorem} \ref{theorem_1}, which only generates
a complexity ordered by $\mathcal{O}(N)$.
For the approximate subproblem \eqref{Q_l_app_1}, the Dinkelbach-like algorithm is proved to exhibit a linear convergence rate~\cite{J_A.Zappone15Energy} and the corresponding convex optimization problem can be efficiently solved by CVX, thus the generated complexity is acceptable in general.

In order to further reduce the complexity of solving the optimization problem for minimizing the network's total energy consumption, we will consider the case of applying the massive MIMO technology at the MBS in the following section. It demonstrates that the complexity of the proposed algorithm can be substantially reduced while  even better performance can be achieved compared to the case with traditional MIMO backhaul.

\section{Massive MIMO Backhaul}\label{sec:Special_MEC_ZF}
In the prior sections, we have studied the synergy of combining edge-central cloud computing with traditional multi-cell MIMO backhaul. Since massive MIMO has been one of the key 5G radio-access technologies, in this section, we further consider the time-division duplex (TDD) massive MIMO aided backhaul in the Rayleigh fading environment, i.e., MBS is equipped with a very large number of antennas and SBSs only use one single transmit antenna ($M \gg N$).

There are two main merits for massive MIMO backhaul transmission: 1) Since SBSs and MBSs are usually still and the backhaul channels will become deterministic, a phenomenon known as ``channel hardening''~\cite{Bjorson_mag_2016,hien_2018_pilot}, and thus the backhaul channel coherence time will be much longer than ever before, which means that the time spent on uplink channel estimation will be much lower. %Some real-time massive MIMO channel measurement works such as~\cite{P_Harris_2018} also demonstrated that the use of massive antennas can mitigate the fast-fade error bursts, and enable much less frequent update of power control in low-mobility environments compared to the single-antenna case (see \cite[Fig.~8]{P_Harris_2018});
2) As shown in~\cite{Marzetta_2010_Nonc}, simple linear processing methods can achieve nearly-optimal performance. As a result, we will consider two linear detection schemes at the MBS, namely MRC and ZF, to  provide low-complexity massive MIMO backhaul solutions.
\subsection{MRC Receiver at the MBS}\label{sec:MRC}
When MRC receiver is applied at the MBS, we consider a lower-bound achievable backhaul rate for tractability, which can well
approximate the exact massive MIMO transmission rate as confirmed in~\cite{ngo2013energy}. %Based on~\cite[Propostion 2]{ngo2013energy},
As such, given the cloud selection decision $\mathbf{\widehat{c}}$, %$\left(\mathbf{c},\mathbf{p}^{\mathrm{u}},\mathbf{w}\right)$
the backhaul related problem \eqref{P1_2_Q} reduces to
\begin{align}\label{P1_2_Q_massiveMIMO_MRC}
&\mathop {\min }\limits_{\mathbf{q} }~\sum\limits_{n{\rm{ = }}1}^N \left( {1{\rm{ - }}{\widehat{c}_n}} \right)
q_n \frac{I_n}{R_n^{\mathrm{b}}(\mathbf{q})}  \\
&~\mathrm{s.t.} ~~\mathrm{\widehat{C}2}: R_n^{\mathrm{b}}(\mathbf{q}) \geq \left(1-\widehat{c}_n\right)\frac{I_n}{\alpha T_n^{\mathrm{edge}}},\ \forall n\in \mathcal{N},\nonumber\\
&~~~~~~~\mathrm{C5}: q_n \geq 0,\ \forall n\in\mathcal{N}, \nonumber
\end{align}
where $q_n$ is the $n$-th SBS's transmit power, $\mathbf{q}=[q_1,\cdots,q_N]$, and
$R_n^{\mathrm{b}}(\mathbf{q})=B^\mathrm{b} \log_2\Big(1+\left(M-1\right)\frac{q_n \beta_n}{\sum\nolimits_{i{\rm{ = }}1{\rm{,}}i \ne n}^N {{q_i}{\beta _i}} +\sigma_n^2}\Big)$, %~\cite{ngo2013energy}
%\begin{align}\label{rate_MRC}
%R_n^{\mathrm{b}}(\mathbf{q})=B^\mathrm{b} \log_2\left(1+\left(M-1\right)\frac{q_n \beta_n}{\sum\nolimits_{i \ne n}^N {{q_i}{\beta _i}} %+\sigma_n^2}\right),
%\end{align} %~\cite{ngo2013energy}
in which $\beta_i$ is the large-scale fading coefficient of the link between SBS $i$ and the MBS~\cite{ngo2013energy}. Problem \eqref{P1_2_Q_massiveMIMO_MRC} is non-convex, but can be equivalent to problem \eqref{Access_power_beam} with $\mathbf{w}_n=1$. Thus, it can be directly solved by using \textbf{Algorithm}~\ref{algorithmic1}. Note that when using \textbf{Algorithm}~\ref{algorithmic1}, SBSs'  initial feasible transmit power vector $\mathbf{q}$ needs to be carefully selected. Here, we assume that the present fractional power control solution applied in 3GPP-LTE~\cite{A_He_2017} can satisfy the constraint $\mathrm{\widehat{C}2}$ in \eqref{P1_2_Q_massiveMIMO_MRC}, i.e., $q_n=(d_n)^{\epsilon \varpi^\mathrm{b}}$, where $d_n$ is the communication distance  between the $n$-th SBS and the MBS,  $\epsilon \in \left[0,1\right]$ is the pathloss compensation factor, and $\varpi^\mathrm{b}$ is the pathloss exponent of the backhaul link. For the special case of full compensation ($\epsilon=1$), the number of MBS's antennas needs to meet
\begin{align}\label{massiveAntenna_MRC}
M \geq 1+\left(N-1\right){\left( 2^{\frac{\left(1-\widehat{c}_n\right)I_n}{B^\mathrm{b} \alpha T_n^{\mathrm{edge}}}}-1\right)}.
\end{align}

\subsection{ZF Receiver at the MBS}\label{sec:ZF_masssive_MIMO}
When ZF receiver is applied at the MBS, we adopt  the corresponding tight lower-bound achievable rate shown in~\cite{ngo2013energy}.  Given the cloud selection decision $\mathbf{\widehat{c}}$, the backhaul related  problem \eqref{P1_2_Q} reduces to the following version
\begin{align}\label{P1_2_Q_massiveMIMO_ZF}
&\mathop {\min }\limits_{\mathbf{q} }~\sum\limits_{n{\rm{ = }}1}^N
 \left( {1{\rm{ - }}{\widehat{c}_n}} \right)  \frac{  q_n I_n} {R_n^{\mathrm{b}}(q_n)} \\
&~\mathrm{s.t.} ~~\mathrm{\widehat{C}2}: R_n^{\mathrm{b}}(q_n) \geq \left(1-\widehat{c}_n\right) \frac{ I_n}{\alpha T_n^{\mathrm{edge}}},\ \forall n \in \mathcal{N}, \nonumber\\
&~~~~~~~\mathrm{C5}: q_n \geq 0,\ \forall n \in \mathcal{N}, \nonumber
\end{align}
where $R_n^{\mathrm{b}}(q_n)=B^\mathrm{b} \log_2\Big(1+\left(M-N\right)\frac{q_n \beta_n}{\sigma_n^2}\Big)$.
Since $\frac{ q_n } {R_n^{\mathrm{b}}(q_n)}$ is an increasing function of $q_n$ ( $\frac{{\partial \big( {\frac{{{q_n}}}{R_n^{\mathrm{b}}(q_n)}} \big)}}{{\partial {q_n}}} \geq 0$), the optimal $q_n^*$ is the minimum value that meets the constraints $\mathrm{\widehat{C}2}$ and $\mathrm{C5}$ in \eqref{P1_2_Q_massiveMIMO_ZF}, i.e.,
\begin{align}\label{optimal_ZF}
{q_n^*}=\frac{2^{\frac{\left(1-\widehat{c}_n\right) I_n}{B^\mathrm{b} \alpha T_n^{\mathrm{edge}}}}-1}{\left(M-N\right)\frac{\beta_n}{\sigma_n^2}}, \ \forall n \in \mathcal{N}.
\end{align}

Based on the above analysis,  when massive MIMO backhaul is employed at the MBS, the solution of problem \eqref{P1} can still be obtained by using the proposed \textbf{Algorithm}~\ref{algorithmic2}, where the optimal SBSs' transmit powers are  given by the solution of problem \eqref{P1_2_Q_massiveMIMO_MRC} for the MRC receiver  or \eqref{optimal_ZF} for the ZF receiver.

In comparison with the case of using traditional MIMO backhaul, the MRC and ZF linear detection schemes for the case with massive MIMO backhaul links can enjoy super-low complexity. For  MRC scheme, the problem \eqref{P1_2_Q_massiveMIMO_MRC} can be effectively solved by  \textbf{Algorithm}~\ref{algorithmic1}, and its computational complexity is with the order of $\mathcal{O}(N)$. For ZF scheme, the closed-form solution of problem \eqref{P1_2_Q_massiveMIMO_ZF} can be directly obtained, and its complexity order is $\mathcal{O}(1)$. Hence, applying the massive MIMO technology at the MBS can significantly facilitate the cooperation between the edge and central cloud by providing easier but more efficient backhaul offloading for UEs  to  access  the central cloud computing services.

\setlength{\tabcolsep}{0.3 pt}\begin{table*}[thb]   % \footnotesize
\centering
\caption{Simulation Parameters}\label{table1}
\begin{tabular}{|l|l|l|}
\hline
~\textbf{Parameter }&~{\textbf{Symbol}} &~{\textbf{Value}} \\
\hline
~Bandwidth for an access or backhaul link\quad\quad\quad\quad\quad\quad\quad\quad\quad\quad\quad &~$B^\mathrm{a}$,$B^\mathrm{b}$ \quad\quad\quad\quad\quad\quad\quad\quad &~10 MHz \quad\quad\quad\quad\quad\quad\quad \\
%\hline
%~Bandwidth for a backhaul link &$B^b$& 10 MHz  \\
\hline
~Noise power spectral density for an access  or backhaul link &~$\sigma_n^2, n\in\mathcal{N}$, ~$\sigma^2$ &~-174 dBm/Hz \\
\hline
~Pathloss exponent for access link &~$\varpi^\mathrm{a}$ &~3.67 \\
\hline
~Pathloss exponent for backhaul link &~$\varpi^\mathrm{b}$ &~2.35 \\
\hline
~Pathloss compensation factor &~$\epsilon$ &~1 \\
\hline
~Radius of the small cell &~$r^\mathrm{a}$ &~50 m  \\
\hline
~Radius of the macro cell &~$r^\mathrm{b}$ &~500 m  \\
\hline
~Number of SBSs/UEs &~$N$ &~6  \\
\hline
~Number of antennas for each SBS &~$L$ &~2  \\
\hline
~UEs' maximum transmit power &~$P_{\max}^{\rm u}$ &~23 dBm  \\
\hline
~Required CPU cycles per bit &~$K_n,~n\in\mathcal{N}$ &~300 cycles/bit \\
%\hline
%~SBS's CPU clock speed &~$f$ &~$6*10^9$ CPU cycles/s \quad\quad\quad\\
\hline
~the effective switched capacitance of the SBSs' processor\quad\quad\quad &~$\varrho_n, ~n\in\mathcal{N}$ &~$10^{-28}$ \\
\hline
~The tolerant value in \textbf{Algorithm} 2\quad\quad\quad &~$\delta$ &~0.1 \\
\hline
\end{tabular}
\end{table*}

\section{Simulation Results}\label{sec:simulation}

In this section, simulation results are presented to evaluate the performance of the proposed algorithms and shed light on the effects of the key parameters including the ratio of energy consumption
between central and edge cloud computing ($\zeta_n=\zeta, n\in\mathcal{N}$), the task size ($I_n=I, n\in\mathcal{N}$), the latency threshold of edge processing
 ($T_{\mathrm{th}}$), the required fraction of edge computing time for backhaul transmission ($\alpha$), and the edge clouds' CPU clock frequency ($f_n=f, n\in\mathcal{N}$).
The performance of some practical schemes are also given as benchmarks, including the ``Edge-cloud-only", ``Central-cloud-only" schemes, and a scheme with fixed cloud selection, denoted as ``Half edge, half central" scheme where half number of UEs choose edge cloud and the other half use central cloud to complete their computing tasks. Besides, the ``Initial feasible solution", representing the case with the initial values setting in Algorithm \ref{algorithmic2}, is also given as a baseline to show the performance improvement of optimizing some system parameters.
Note that the performance indicators (the total energy consumption and the percentage of UEs that select edge cloud computing) shown in the following figures are averaged over 500 independent channel realizations.
All the small-scale fading channel coefficients follow independent and identically complex Gaussian distribution with zero mean and unit variance.
The pathloss between SBSs and UEs  and between MBS and SBSs are respectively set as $140.7+36.7\log_{10}d$(km) and $100.7+23.5\log_{10}d$(km) according to 3GPP TR 36.814~\cite{3GPP_2017}, where $d$ is the distance between two nodes.
The other basic simulation parameters are listed in Table~\ref{table1}.

\subsection{Improvement with Traditional MIMO Backhaul}\label{sec:traditional MIMO}
In this subsection, numerical results for the integrated edge and central cloud computing system with traditional MIMO backhaul are presented in comparison with the benchmarks mentioned before. These results can {properly demonstrate} the performance enhancement of using the proposed algorithm through jointly optimizing the key system parameters including cloud selection decision, UEs' transmit powers, SBSs' receive beamformers and transmit covariance matrices.

%\begin{figure}[tbp]%[t!]
%\centering
%%\includegraphics[width=3.3 in]
%\includegraphics[scale=0.55]{New_Energy_zeta.eps}
%\caption{The total energy consumption of the system with traditional MIMO backhaul versus $\zeta$: $M=16$, $T=0.3$ s, $\alpha=0.1$, $I_n=5$ Mbits.}
%\label{Energy_zeta_MIMO}
%\end{figure}

%\begin{figure}[tbp]%[t!]
%\centering
%%\includegraphics[width=3.3 in]
%\includegraphics[scale=0.55]{Energy_zeta_0912.eps}
%\caption{The total energy consumption of the system with traditional MIMO backhaul versus $\zeta$: $M=16$, $T=0.3$ s, $\alpha=0.1$, $I_n=5$ Mbits.}
%\label{Energy_zeta_MIMO1}
%\end{figure}

\begin{figure}[tbp]%[t!]
\centering
\includegraphics[scale=0.55]{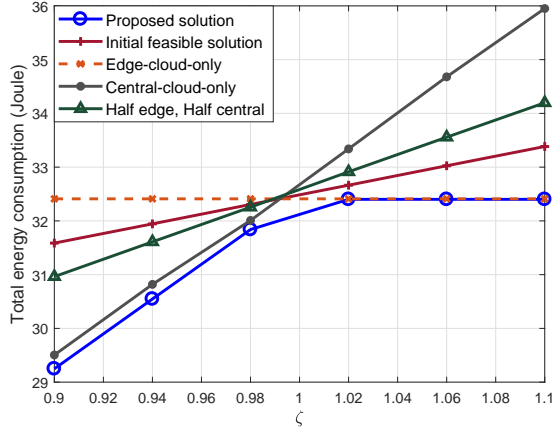}%{Energy_zeta_09108.eps}
\caption{The total energy consumption of the system with traditional MIMO backhaul versus the uniform computing energy ratio $\zeta$: $M=16$, $T_{\mathrm{th}}=0.3$ s, $\alpha=0.1$, $I=I_n=5$ Mbits, $f=f_n=6$ GHz for $n\in\mathcal{N}$.} %{$\zeta=\zeta_n$ for $n\in\mathcal{N}$}
\label{Energy_zeta_MIMO}
\end{figure}

Fig.~\ref{Energy_zeta_MIMO} shows the effect of the uniform computing energy ratio $\zeta=\zeta_n, n\in\mathcal{N}$ on the total energy consumption of the system.
We see that the energy consumption of all the schemes are non-decreasing functions of $\zeta$, due to the fact that the energy cost of central cloud computing increases with $\zeta$. It is confirmed that the proposed solution outperforms all the baselines, i.e., the energy cost can be significantly reduced.
The performance improvement is particularly noticeable compared with the Edge-cloud-only scheme in the range of $\zeta<1$,
the traditional Central-cloud-only scheme in the range of $\zeta>1$, and the ``Half edge, half central'' scheme in the whole
range of $\zeta$.  In addition, the proposed solution also consumes much less energy than the initial
feasible solution, demonstrating the performance enhancement of jointly optimizing the system parameters.

\begin{figure}
\centering
\includegraphics[scale=0.55]{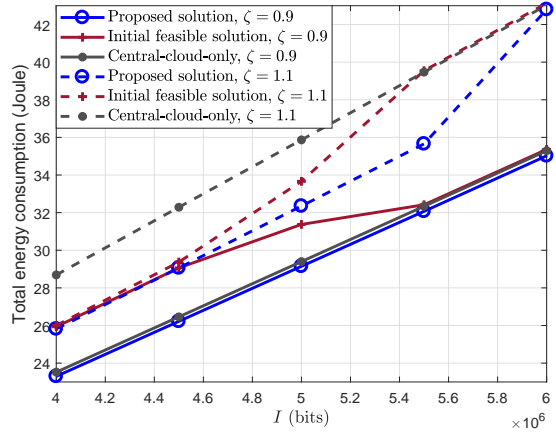}
\caption{The total energy consumption of the system with traditional MIMO backhaul versus the uniform task size $I$: $M=16$, $T_{\mathrm{th}}=0.3$ s,  $f=f_n=6$ GHz for $n\in\mathcal{N}$, $\alpha=0.1$.}
\label{Energy_I_MIMO}
\end{figure}

Fig.~\ref{Energy_I_MIMO} depicts the total energy consumption of the system versus the uniform task sizes $I=I_n, n\in\mathcal{N}$ for the cases of $\zeta=0.9$ and $\zeta=1.1$. It is easy to understand that computing more input data consumes more energy, and thus the energy cost of each scheme increases with $I$. Again, we see that the proposed solution is superior to the baseline solutions in all the cases. For the case of $\zeta=0.9$, the performance of the Central-cloud-only solution is very close to the proposed one since central cloud is dominant in this case, i.e., more UEs tend to use central cloud computing for saving energy. For the case of $\zeta=1.1$, the advantage of the proposed scheme becomes more obvious compared with the baselines, and actually this case is more common in practice since the central cloud tends to consume more energy for computing because of the higher CPU frequency.
We observe that the results of the proposed solution approach to those of the Central-cloud-only solution when $I$ becomes large, indicating that more UEs tend to select the central cloud for computing, i.e., central cloud computing plays an important role in dealing with relatively large tasks.
The reason is that when the task size is large, the edge processing latency  constraint $\mathrm{C3}$ of problem \eqref{P1} may be no longer satisfied due to the limited edge computing capability, and central cloud has to be chosen for computation.

\begin{figure}[tbp]%[t!]
\centering
\includegraphics[scale=0.55]{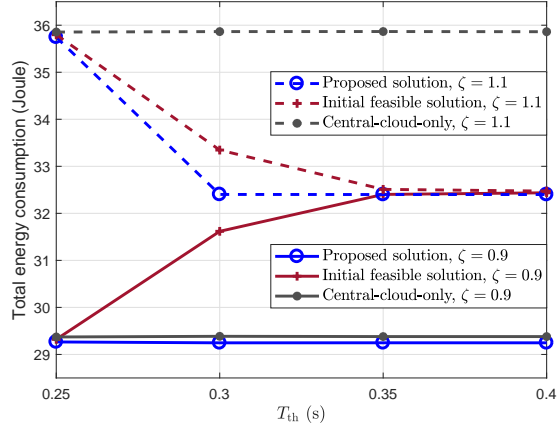}%{Energy_T_traditional_HIC.eps}
\caption{The total energy consumption of the system with traditional MIMO backhaul versus the latency threshold of edge processing  $T_{\mathrm{th}}$: $M=16$, $I=I_n=5$ Mbits, $f=f_n=6$ GHz for $n\in\mathcal{N}$, $\alpha=0.1$.}
\label{Energy_T_MIMO}
\end{figure}
Fig.~\ref{Energy_T_MIMO} shows  the total energy consumption of the system varying with the latency threshold of edge processing for the cases of $\zeta=0.9$ and $\zeta=1.1$. It is seen that the proposed solution is a non-increasing function of $T_{\mathrm{th}}$ and outperforms the baselines in both cases. The Central-cloud-only solution is insensitive to $T_{\mathrm{th}}$, and its performance is almost
invariant
thanks to its super computing capability for low computing latency. Note that all the solutions consume almost same amount of energy
when $T_{\mathrm{th}}$ is small, e.g., $T_{\mathrm{th}}=0.25~\mathrm{s}$ in this figure.  The reason is  that the edge processing latency constraint  $\mathrm{C3}$ cannot be met and only central cloud computing can be employed to
satisfy the latency constraints.
For the case of $\zeta=0.9$, the performance gap between the proposed solution and the Central-cloud-only  is small since central cloud computing is dominant, and both solutions performs better than
the Initial feasible solution. It is interesting to note that the Initial feasible solution is an increasing function of
$T_{\mathrm{th}}\in [0.25,0.4]~\mathrm{s}$ when $\zeta=0.9$. This is because the edge cloud computing becomes more feasible as $T_{\mathrm{th}}$ increases, and the initial solution allowing more UEs to choose edge cloud for computing while in fact central cloud computing saves more energy, which indicates the importance of optimizing cloud selection in improving the system performance.
For the case of $\zeta=1.1$,  the consumed energy of the proposed solution decreases with $T_{\mathrm{th}}$ since more UEs are allowed to choose the energy-efficient edge cloud computing for large $T_{\mathrm{th}}$.

\subsection{Benefits of Massive MIMO Backhaul}\label{sec:massive MIMO}
In this subsection, we mainly illustrate the performance of the  considered heterogeneous  edge/central cloud computing system with massive MIMO backhaul, to confirm the benefits of equipping massive antennas at the MBS in improving the system performance. Here, we focus on MRC and ZF beamforming at the MBS, as studied in Section IV.

\begin{figure}[tbp]%[t!]
\centering
\includegraphics[scale=0.55]{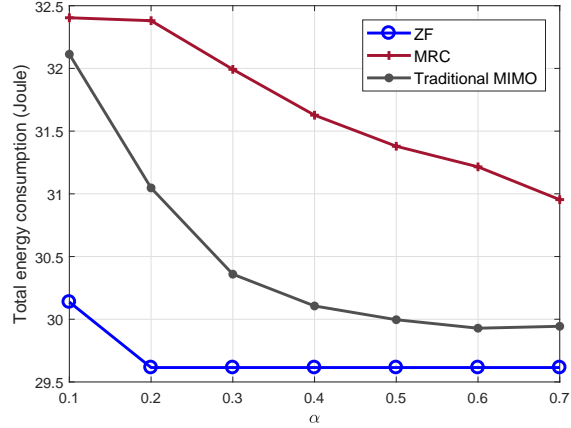}
\caption{The total energy consumption of the system versus  $\alpha$: $M=128$ for massive MIMO backhaul,  $M=8$ for traditional
MIMO backhaul,  $T_{\mathrm{th}}=0.3$ s, $\zeta=\zeta_n=0.9$, $I=I_n=5$ Mbits, $f=f_n=6$ GHz for $n\in\mathcal{N}$.} %, $P_{\mathrm{max}}^u=23$dBm
\label{Energy_alpha_mass_traditional}
\end{figure}

\begin{figure}[tbp]%[t!]
\centering
\includegraphics[scale=0.55]{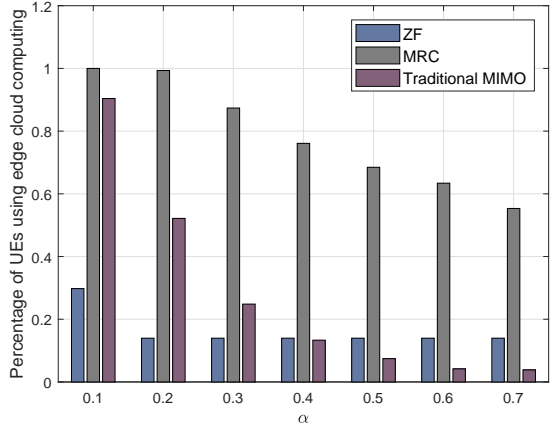}%{cn_alpha_edge_mass_traditional_100.eps}
\caption{The percentage of UEs that select edge cloud computing versus $\alpha$:  $M=128$ for massive MIMO backhaul,  $M=8$ for traditional MIMO backhaul,  $T_{\mathrm{th}}=0.3$ s, $\zeta=\zeta_n=0.9$, $I=I_n=5$ Mbits, $f=f_n=6$ GHz for $n\in\mathcal{N}$.}
\label{cn_alpha_mass_traditional}
\end{figure}
Fig.~\ref{Energy_alpha_mass_traditional} and Fig.~\ref{cn_alpha_mass_traditional}  depict the total energy consumption and the corresponding percentage of UEs that select edge cloud for computing versus $\alpha$, respectively. It is seen from Fig.~\ref{Energy_alpha_mass_traditional} that the energy consumption of each  scheme decreases with $\alpha$ since less power will be consumed for backhaul transmission with a higher $\alpha$  according to the backhaul latency constraint  $\mathrm{C2}$ of problem \eqref{P1}.
This result is also reflected by Fig.~\ref{cn_alpha_mass_traditional} where the percentage of UEs using edge cloud computing decreases, which means that more UEs choose to use central cloud for computing as $\alpha$ increases so as to save more energy. Obviously, the energy consumed by the ZF scheme is less than that of the MRC scheme and the solution with traditional MIMO backhaul, which demonstrates the benefits of using ZF beamforming and large antenna arrays at the MBS.
Moreover, for the ZF scheme, the percentage of UEs using edge cloud is lower than that of the MRC and traditional MIMO schemes when $\alpha<0.4$. In contrast, the MRC scheme only uses the edge cloud for computing when $\alpha\leq0.2$. This is because the backhaul latency constraint  $\mathrm{C2}$ in~\eqref{P1} for central cloud processing cannot be satisfied with a small $\alpha$ when MRC receiver is adopted at the MBS due to the inter-SBS interference. Based on these two figures, we  see that the consumed energy of the ZF scheme as well as the corresponding percentage of UEs served by edge cloud  decrease very slowly, and is almost unchanged for $\alpha\geq0.2$, which further indicates that the ZF scheme can provide more stable and higher-speed backhaul transmission for computation tasks offloading.
\begin{figure}[tbp]%[t!]
\centering
\includegraphics[scale=0.55]{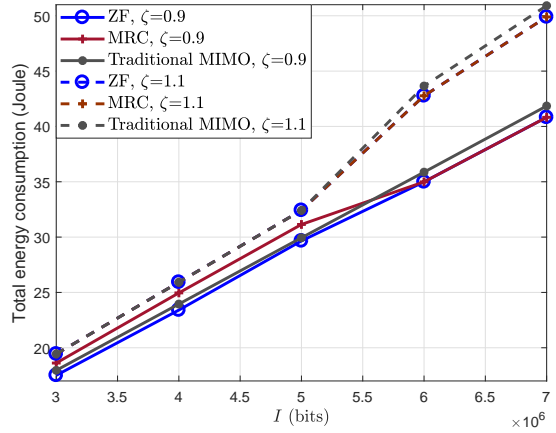}%{New_Energy_I_mass_traditional_100.eps}
\caption{The total energy consumption of the system versus the uniform  task size $I$: $M=128$ for massive MIMO backhaul,
$M=8$ for traditional MIMO backhaul, $T_{\mathrm{th}}=0.3$ s, $f=f_n=6$ GHz for $n\in\mathcal{N}$, $\alpha=0.6$.} %, $P_{\mathrm{max}}^u=23$ dBm.
\label{Energy_I_mass_traditional}
\end{figure}

\begin{figure}[tb]
\centering
\subfigure[$M=128$ for massive MIMO backhaul,  $M=8$ for traditional MIMO backhaul, $T_{\mathrm{th}}=0.3$ s, $\alpha=0.6$, $\zeta=\zeta_n=0.9$ for $n\in\mathcal{N}$.] {\label{Energy_f_mass_traditional_100}
\includegraphics[scale=0.52]{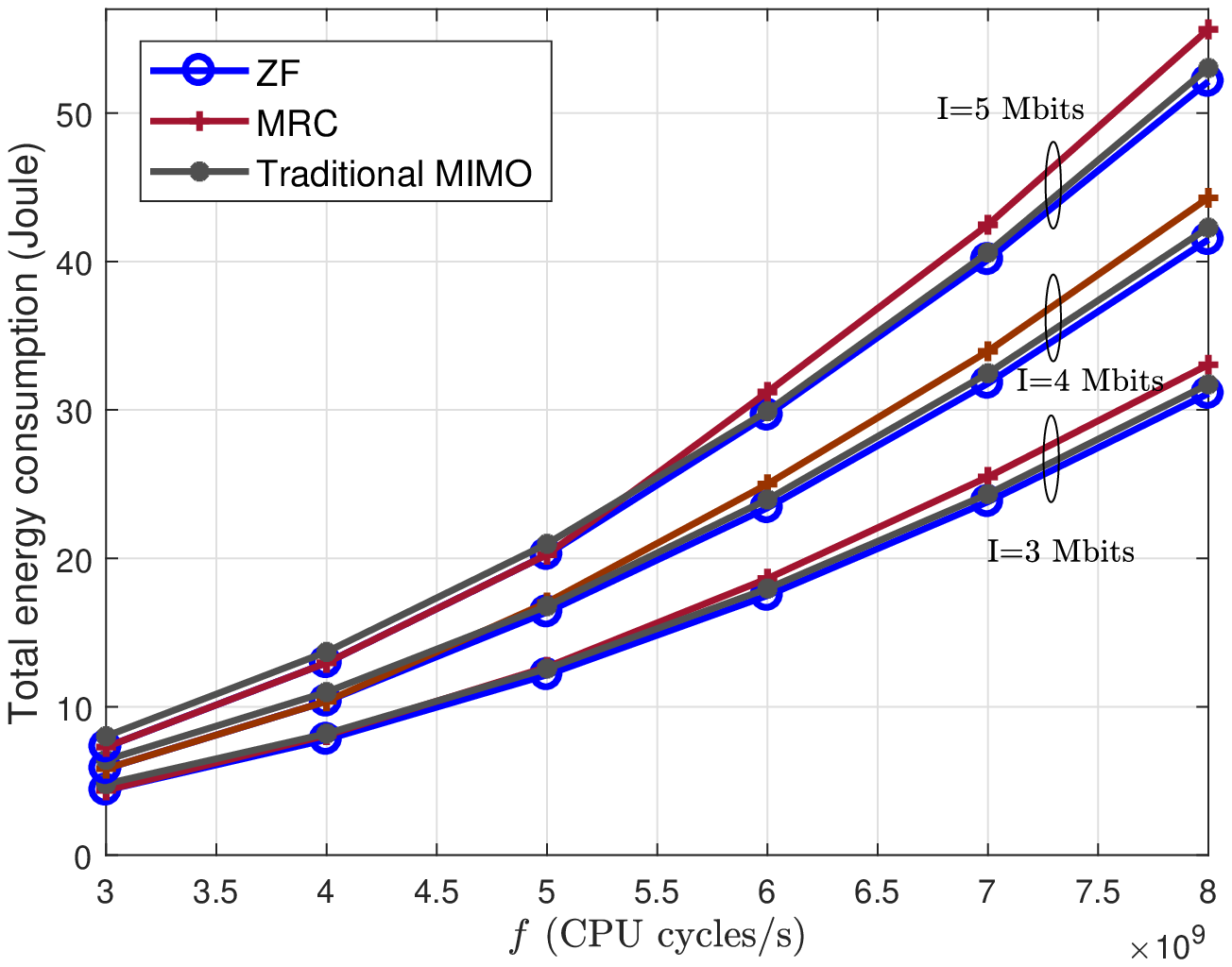}
}
\subfigure[ $M=128$ for massive MIMO backhaul,  $M=8$ for traditional MIMO backhaul,  $T_{\mathrm{th}}=0.3$ s, $\alpha=0.6$, $\zeta=\zeta_n=1.5$ for $n\in\mathcal{N}$.]{
\label{1Energy_f_mass_traditional_100}
  \includegraphics[scale=0.52]{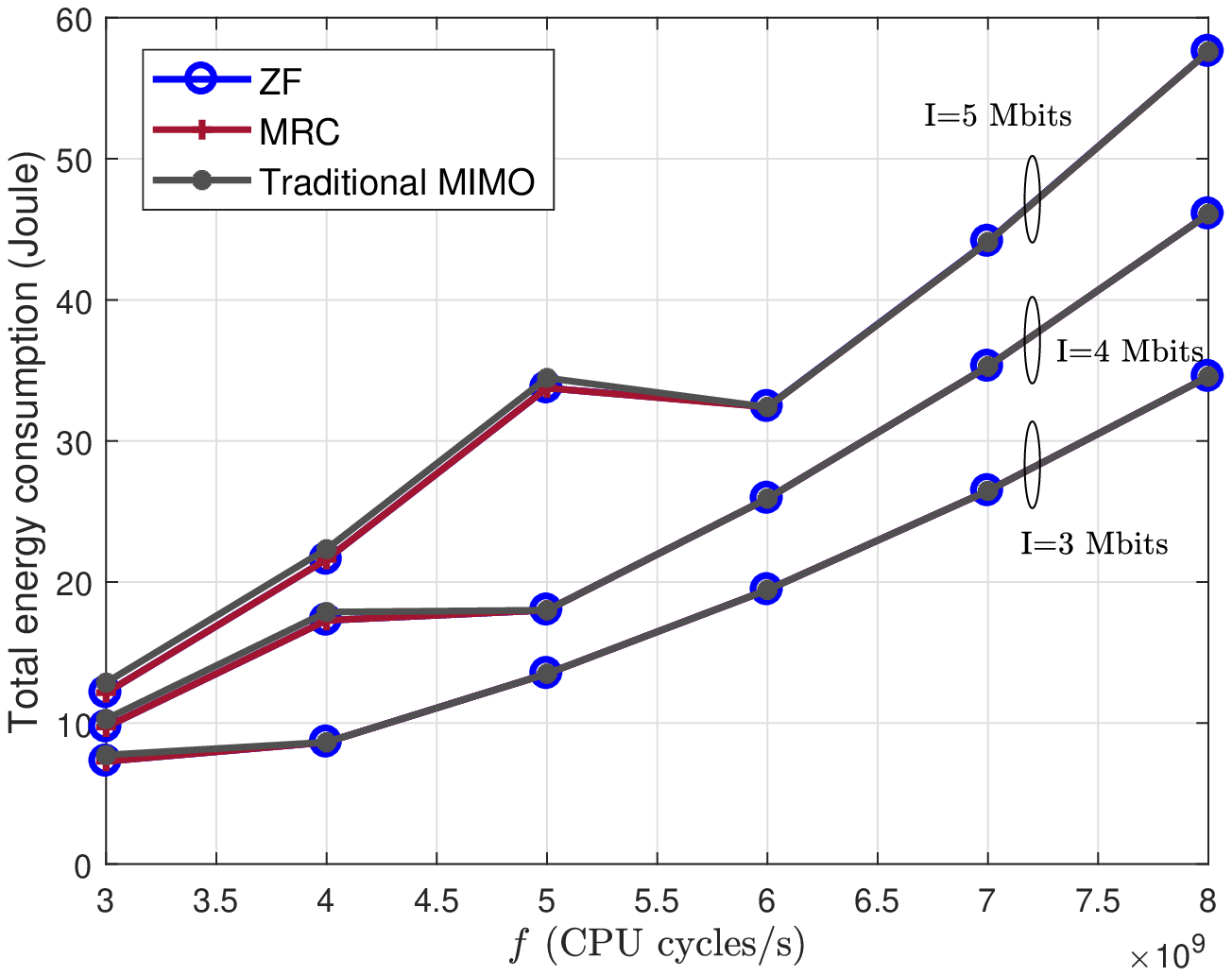}}
 \caption{The total energy consumption of the system versus SBSs' uniform CPU frequency $f$.}
\label{Fig8}
\end{figure}

Fig.~\ref{Energy_I_mass_traditional} shows the total energy consumption of the system versus the uniform task size $I$ for the cases of $\zeta=0.9$ and $\zeta=1.1$.  Similar to Fig.~\ref{Energy_I_MIMO}, all the curves increase with $I$ as expected. The ZF scheme outperforms the MRC scheme and the traditional MIMO scheme. For the case of $\zeta=0.9$, the ZF scheme and the traditional MIMO scheme are dominated by central cloud computing, while the MRC scheme experiences a gradual transition
from edge-cloud-dominant to central-cloud-dominant  and more UEs choose to use central cloud for computing so as to satisfy the
processing latency constraint as well as saving energy. For the case of $\zeta=1.1$, all the schemes are edge-cloud dominant when $I\leq5$ Mbits, and then gradually become central-cloud-dominant as $I$ increases. It is confirmed that the ZF scheme with massive MIMO backhaul has the advantage of handling the computation-intensive tasks.

Fig.~\ref{Energy_f_mass_traditional_100} and Fig.~\ref{1Energy_f_mass_traditional_100} depict the total energy consumption of the system versus the edge clouds' uniform CPU clock frequency $f=f_n, n\in\mathcal{N}$ in the case of $\zeta=0.9$ and $\zeta=1.5$, respectively.
According to these two figures, we see that the effect of $f$ is heavily reliant on both the computing task size $I$ and $\zeta$.
When $I$ is not large and $\zeta <1$, network's energy consumption may increase with $f$ as shown in Fig.~\ref{Energy_f_mass_traditional_100}, where the curves of all the schemes increase with $f$ and the increasing rates become higher when enlarging $I$. This is due to the fact that when $I$ is not large and $\zeta <1$,  the energy consumption of the central cloud computing plays a dominant role in contributing to the total energy consumption.
In this case, the advantage of using ZF scheme  becomes more obvious as $f$ grows large.
However, when $\zeta >1$,  network's energy consumption may decrease with $f$ in certain scenario  as shown in Fig.~\ref{1Energy_f_mass_traditional_100}, where there is an obvious decrease as $f\in[5,6]\times10^9~\mathrm{cycles/s}$ in the case of $I=5$ Mbits. The reason is that when $f$ is small, e.g., less than $4\times10^9~\mathrm{cycles/s}$ in
Fig.~\ref{1Energy_f_mass_traditional_100},
 the edge processing latency constraint  $\mathrm{C3}$  may be not satisfied and central cloud computing becomes the only option.  As $f$ increases, edge cloud computing becomes  feasible for more UEs to save energy, and the total energy cost will decrease. In addition,
it is seen from Fig.~\ref{1Energy_f_mass_traditional_100} that the energy consumption of the three considered schemes are very close since the edge cloud computing is dominant.

\section{Conclusion}\label{sec:conclusion}
In this paper, we studied the joint design of computing services when edge cloud computing and central cloud computing coexist in a two-tier HetNet with MIMO or massive MIMO self-backhaul. By jointly optimizing the cloud selection, the UEs' transmit powers, the SBSs' receive beamforming vectors and the transmit covariance matrices, the network's energy consumption for task offloading and computation can be minimized while meeting both the edge processing and central processing (backhaul) latency constraints. An iterative algorithm was proposed to solve the formulated non-convex mixed-integer optimization problem, which can ensure the convergence and that better performance can be achieved than any existing feasible solutions. The simulation results have further confirmed  that the proposed solution can greatly enhance the system performance, especially comparing with the edge-cloud-only and central-cloud-only computing schemes, indicating the value of cooperation between edge and central clouds.  Moreover, we showed that the massive MIMO backhaul can largely decrease the system complexity while achieving even better performance.

%\newpage

\section*{Appendix A: Proof of {Lemma}~\ref{lemma1}}
\label{App:lemma_1}
\renewcommand{\theequation}{A.\arabic{equation}}
\setcounter{equation}{0}
Based on problem \eqref{Access_power_beam}, we can easily find that each SBS's receive beamformer $\mathbf{w}_n$ aims to maximize the SINR, i.e.,
\begin{align}\label{object_beam}
\mathop {\max }\limits_{\mathbf{w}_n}~\gamma_n^{(\mathrm{a})}(\mathbf{p}^{\mathrm{u}},\mathbf{w}_n)
\end{align}
Problem \eqref{object_beam} can be rewritten as
\begin{align}\label{object_beam_12}
\mathop {\max }\limits_{\mathbf{w}_n}~\frac{\mathbf{w}_n^H \mathbf{\Omega}_n   \mathbf{w}_n}{\mathbf{w}_n^H \mathbf{\Omega}_{-n} \mathbf{w}_n }.
\end{align}
Note that \eqref{object_beam_12} is a generalized eigenvector problem and the optimal $\mathbf{w}_n^*$ is the corresponding  eigenvector associated with the largest eigenvalue of the matrix $\left(\mathbf{\Omega}_{-n}\right)^{-1} \mathbf{\Omega}_n$. Thus, we obtain the result in~\eqref{w_beam}.
\section*{Appendix B: Proof of {Lemma}~\ref{lemma2}}
\label{App:lemma_2}
\renewcommand{\theequation}{B.\arabic{equation}}
\setcounter{equation}{0}

The Lagrange function of problem \eqref{P1_max_sub1} is
\begin{align}\label{B1}
\mathcal{L}\left(\mathbf{p}^{\mathrm{u}},\mathbf{t}, \mathbf{\lambda}, \mathbf{\mu},\mathbf{\nu}\right)&=\sum\limits_{n{\rm{ = }}1}^N I_n t_n +\sum\limits_{n{\rm{ = }}1}^N \lambda_n
\left(p_n^{\mathrm{u}}-t_n R_n^{\mathrm{a}}(\mathbf{p}^{\mathrm{u}},\mathbf{w}_n) \right) \nonumber\\
&+\sum\limits_{n{\rm{ = }}1}^N \mu_n \left(\tau-\gamma_n^{\mathrm{a}}(\mathbf{p}^{\mathrm{u}},\mathbf{w}_n)\right)
+\sum\limits_{n{\rm{ = }}1}^N \nu_n\left(p_n^{\mathrm{u}}-P_{\max}^{\rm u}\right),
\end{align}
where $\left\{\lambda_n,\mu_n,\nu_n\right\}_{n=1}^N$ are non-negative Lagrange multipliers. Based on the definition of KKT conditions, we have
%\begin{align}
%  &\frac{\partial \mathcal{L}}{\partial p_n^{\mathrm{u}}}=\lambda_n
%  -\lambda_n t_n\frac{\partial R_n^{\mathrm{a}}}{\partial p_n^{\mathrm{u}}}
%  -\mu_n\frac{\partial \gamma_n^\mathrm{a}}{\partial p_n^{\mathrm{u}}}+\nu_n %\nonumber\\
%  %&\quad\quad
%  -\sum_{j\neq n}^N\lambda_j t_j\frac{\partial R_j^{\mathrm{a}} }{\partial p_n^{\mathrm{u}}}
% -\sum_{j\neq n}^N\mu_j\frac{\partial \gamma_j^\mathrm{a}}{\partial p_n^{\mathrm{u}}}=0, \label{b2_1}\\
%  & \frac{\partial \mathcal{L}}{\partial t_n }=I_n-\lambda_n R_n^{\mathrm{a}}=0, \label{b2_2}\\
% &  \lambda_n\left(p_n^{\mathrm{u}}-t_n R_n^{\mathrm{a}}\right)=0, \label{b2_3}\\
% & \mu_n \left(\tau-\gamma_n^{\mathrm{a}}\right)=0,  \label{b2_4}\\
 % & \nu_n\left(p_n^{\mathrm{u}}-P_{\max}^{\rm u}\right)=0.  \label{b2_5}
%\end{align}
\begin{align}
  &\frac{\partial \mathcal{L}}{\partial p_n^{\mathrm{u}}}=\lambda_n
  -\lambda_n t_n\frac{\partial R_n^{\mathrm{a}}}{\partial p_n^{\mathrm{u}}}
  -\mu_n\frac{\partial \gamma_n^\mathrm{a}}{\partial p_n^{\mathrm{u}}}+\nu_n
 -\sum_{j\neq n}^N\lambda_j t_j\frac{\partial R_j^{\mathrm{a}} }{\partial p_n^{\mathrm{u}}}
 -\sum_{j\neq n}^N\mu_j\frac{\partial \gamma_j^\mathrm{a}}{\partial p_n^{\mathrm{u}}}=0, \label{b2_1}\\
  & \frac{\partial \mathcal{L}}{\partial t_n }=I_n-\lambda_n R_n^{\mathrm{a}}=0, \label{b2_2}\\
 &  \lambda_n\left(p_n^{\mathrm{u}}-t_n R_n^{\mathrm{a}}\right)=0, \label{b2_3}\\
 & \mu_n \left(\tau-\gamma_n^{\mathrm{a}}\right)=0,  \label{b2_4}\\
  & \nu_n\left(p_n^{\mathrm{u}}-P_{\max}^{\rm u}\right)=0.  \label{b2_5}
\end{align}
In \eqref{b2_1}, we have $\frac{\partial R_j^{\mathrm{a}} }{\partial p_n^{\mathrm{u}}}=-
\frac{B^\mathrm{a}}{\ln2}\frac{ \left(\gamma_j^{\mathrm{a}} \right)^2 |\mathbf{w}_j^H{\mathbf{h}_{n,j}^{\rm a}}|^2}{p_j^{\mathrm{u}}
|\mathbf{w}_j^H{\mathbf{h}_{j,j}^{\rm a}}|^2 \left(1+\gamma_j^{\mathrm{a}}\right)}$, and $\frac{\partial \gamma_j^{\mathrm{a}} }
{\partial p_n^{\mathrm{u}}}=-\frac{ \left(\gamma_j^{\mathrm{a}} \right)^2 |\mathbf{w}_j^H {\mathbf{h}_{n,j}^{\rm a}}|^2}{p_j^{\mathrm{u}}
 |\mathbf{w}_j^H{\mathbf{h}_{j,j}^{\rm a}}|^2 }$.
Based on \eqref{b2_1}--\eqref{b2_5}, we observe that the $N$ subproblems shown in \eqref{P1_max_sub1_sub2} has the same KKT conditions with problem \eqref{P1_max_sub1}. In other words, problems \eqref{P1_max_sub1} and  \eqref{P1_max_sub1_sub2} have the same optimal solution. In addition, since $R_n^{\mathrm{a}}>0$, we have $\lambda_n=\frac{I_n}{R_n^{\mathrm{a}}}$ based on \eqref{b2_2}, and $t_n=\frac{p_n^{\mathrm{u}}}{R_n^{\mathrm{a}}}$ based on \eqref{b2_3}. Likewise, by considering the KKT conditions of $N$ subproblems in \eqref{P1_max_sub1_sub2}, we find that they are identical to those shown in \eqref{b2_1}--\eqref{b2_5}.

\section*{Appendix C: Proof of {Theorem}~\ref{theorem_1}}
\label{App:theo_1}
\renewcommand{\theequation}{C.\arabic{equation}}
\setcounter{equation}{0}

Based on \eqref{b2_1}, \eqref{b2_4} and \eqref{b2_5} of Appendix B, KKT conditions for subproblem \eqref{P1_max_sub1_sub2} is given by
\begin{align}
&\lambda_n+M_n-  \frac{B_\mathrm{a}}{\ln2}\frac{ \lambda_n t_n \Lambda_n }{ 1+\gamma_n^{\mathrm{a}} }-\mu_n \Lambda_n +\nu_n=0, \label{C11_1} \\
& \mu_n \left(\tau-\gamma_n^{\mathrm{a}}\right)=0,  \label{C11_2}\\
& \nu_n\left(p_n^{\mathrm{u}}-P_{\max}^{\rm u}\right)=0,  \label{C11_3}
\end{align}
where $\Lambda_n=\frac{|\mathbf{w}_n^H{\mathbf{h}_{n,n}^{\rm a}}|^2}{\sum_{i=1,i\neq n}^{N} p_i^{\mathrm{u}}|
\mathbf{w}_n^H {\mathbf{h}_{i,n}^{\rm a}}|^2+|\mathbf{w}_n^H\mathbf{n}_n|^2}$.
From \eqref{C11_1} and the definition of $\gamma_n^{\mathrm{a}}=p_n^{\mathrm{u}}\Lambda_n$ in \eqref{eq:SINR1}, we see that the optimal $p_n^{\mathrm{u}*}$ meets
\begin{align}\label{optimal_p}
p_n^{\mathrm{u}*}=\frac{B_\mathrm{a}}{\ln2} \frac{\lambda_n t_n }{\lambda_n+M_n-\mu_n^* \Lambda_n +\nu_n^*}-\frac{1}{\Lambda_n},
\end{align}
where $\mu_n^*$ and $\nu_n^*$ satisfy the KKT conditions   \eqref{C11_2} and \eqref{C11_3}, respectively. %, as well as \eqref{b2_4} and \eqref{b2_5}.
To explicitly obtain $\{p_n^{\mathrm{u}*},\mu_n^*,\nu_n^*\}$, we need to consider the following  cases:
\begin{itemize}
  \item Case 1: When $p_n^{\mathrm{u}*} \in \left(\frac{\tau}{\Lambda_n},P_{\max}^{\rm u}\right)$, $\mu_n^*=\nu_n^*=0$ according to \eqref{C11_2} and \eqref{C11_3}. In this case, $p_n^{\mathrm{u}*}=G_n$ with $G_n=\frac{B_\mathrm{a}}{\ln2} \frac{\lambda_n t_n }{\lambda_n+M_n}-\frac{1}{\Lambda_n}$ according to \eqref{optimal_p}. Therefore, if $G_n \in \left[\frac{\tau}{\Lambda_n},P_{\max}^{\rm u}\right]$, $p_n^{\mathrm{u}*}=G_n$ and $\mu_n^*=\nu_n^*=0$.

  \item Case 2: If $G_n<\frac{\tau}{\Lambda_n}$, it is seen from \eqref{optimal_p} that $\mu^*>0$. In this case, $p_n^{\mathrm{u}*}=\frac{\tau}{\Lambda_n}$ and $\nu_n^*=0$ according to \eqref{C11_2} and \eqref{C11_3}. Substituting $p_n^{\mathrm{u}*}=\frac{\tau}{\Lambda_n}$ and $\nu_n^*=0$ into \eqref{optimal_p}, we obtain $\mu_n^*=\frac{\lambda_n+M_n}{\Lambda_n}-\frac{B_\mathrm{a}}{\ln2} \frac{\lambda_n t_n}{\tau+1}$

  \item Case 3: If $G_n>P_{\max}^{\rm u}$, it is seen from \eqref{optimal_p} that $\nu_n^*>0$. In this case, $p_n^{\mathrm{u}*}=P_{\max}^{\rm u}$ and $\mu_n^*=0$ according to \eqref{C11_3} and \eqref{C11_2}. Substituting $p_n^{\mathrm{u}*}=P_{\max}^{\rm u}$ and $\mu_n^*=0$ into \eqref{optimal_p}, we obtain $\nu_n^*=\frac{B_\mathrm{a}}{\ln2} \frac{\lambda_n t_n }{P_{\max}^{\rm u}+1/\Lambda_n}  -\lambda_n-M_n$.
\end{itemize}
Thus, we get the optimal $\{p_n^{\mathrm{u}*},\mu^*,\nu_n^*\}$ shown in \textbf{Theorem}~\ref{theorem_1}.

\bibliographystyle{IEEEtran}
%\bibliography{Hybrid_Computing}

\end{document}